\documentclass[letterpaper,journal]{IEEEtran}

\PassOptionsToPackage{compatibility=false}{caption}

\usepackage[caption=false,font=normalsize,labelfont=sf,textfont=sf]{subfig}

\usepackage{silence}
\usepackage{graphicx}

\usepackage{setspace}
\usepackage{amsmath,amssymb,amsfonts}
\usepackage{amsthm}
\usepackage{hyperref}
\usepackage{xcolor}
\usepackage{cite}
\usepackage{float}
\usepackage{stfloats}
\usepackage{comment}
\usepackage{balance}
\usepackage[ruled,vlined]{algorithm2e}
\usepackage{booktabs}
\usepackage{caption}    
\usepackage{tabularx}

\def\a{{\mathbf a}}

\def\s{{\mathbf s}}
\def\z{{\mathbf z}}
\def\n{{\mathbf n}}
\def\v{{\mathbf v}}
\def\x{{\mathbf x}}
\def\y{{\mathbf y}}
\def\X{{\mathbf X}}
\def\Y{{\mathbf Y}}
\def\R{{\mathbf R}}
\def\P{{\mathbf P}}
\def\F{{\mathbf F}}
\def\G{{\mathbf G}}
\def\Q{{\mathbf Q}}
\def\H{{\mathbf H}}
\def\W{{\mathbf W}}
\def\T{{\mathbf T}}
\def\A{{\mathbf A}}
\def\B{{\mathbf B}}

\def\D{{\mathbf D}}
\def\M{{\mathbf M}}

\def\U{{\mathbf U}}
\def\V{{\mathbf V}}
\def\I{{\mathbf I}}

\def\S{{\mathbf S}}

\newcommand\independent{\protect\mathpalette{\protect\independenT}{\perp}}
\def\independenT#1#2{\mathrel{\rlap{$#1#2$}\mkern2mu{#1#2}}}

\newtheorem{theorem}{Theorem}

\theoremstyle{definition}
\newtheorem{remark}{Remark}

\begin{document}

\title{Game-Theoretic Latent Space Alignment for Multi-user Semantic MIMO Communications\vspace{.3cm}}



\author{Giuseppe Di Poce, Mattia Merluzzi, Emilio Calvanese Strinati, and Paolo Di Lorenzo
       
\thanks{Giuseppe Di \!Poce, Mattia Merluzzi and Emilio \!Calvanese \!Strinati are with
CEA Leti, University Grenoble Alpes, 38000 Grenoble, France (e-mail: \{giuseppe.dipoce, mattia.merluzzi,
emilio.calvanese-strinati\} @cea.fr). Paolo Di Lorenzo is with the Department
of Information Engineering, Electronics, and Telecommunications (DIET),
Sapienza University of Rome, 00184 Rome, Italy, and also with CNIT, Parma,
Italy (e-mail: paolo.dilorenzo@uniroma1.it). This work has been founded by French government (2030 ANR, ref. 22-PEFT-0010)
 and by the SNS JU project \!6G-GOALS\! under the EU’s Horizon program Grant
Agreement No \!101139232.
An initial conference work 
appears in \!\cite{di2026distributed}.}
\thanks{Manuscript received May 21, 2026; 
} \vspace{-.5cm}}



\maketitle
\begin{abstract}
Semantic communications enable AI-native wireless systems by mapping raw data into compressed task-oriented latent representations. However, independently trained agents often rely on heterogeneous latent spaces and background knowledge, leading to semantic mismatch that degrades mutual understanding and downstream task execution, especially in interference-limited multi-user wireless networks. This paper investigates distributed latent-space alignment in multi-user semantic MIMO interference networks with cognitive radio constraints. We consider primary users and semantic-aware secondary users sharing the same wireless resources, where secondary agents must simultaneously mitigate interference and align heterogeneous semantic representations. To address this problem, we formulate semantic alignment as a non-cooperative game and derive a closed-form solution for the joint optimization of linear semantic MIMO transceivers under power and interference constraints. Exploiting the structure of the problem, we recast the original matrix-valued optimization into a lower-dimensional power-allocation game, leading to an iterative semantic water-filling algorithm. We establish sufficient conditions for existence, uniqueness, and global convergence to a Nash equilibrium, explicitly relating semantic alignment properties and physical-channel interactions. Numerical results assess the performance of the proposed framework, revealing key trade-offs among semantic compression, task performance, and hierarchical spectrum access.
\end{abstract}


\section{Introduction}
\IEEEPARstart{C}{ommunication}
systems have historically been engineered in accordance with the nature of the information being conveyed and the operational requirements of the target application. 
Information theory, rooted in Shannon’s seminal work~\cite{shannon1948mathematical}, establishes the fundamental limits of reliable communication and data compression for point-to-point systems, thereby defining the theoretical foundations and ultimate performance bounds of communication systems~\cite{cover1999elements}.
Modulation and source coding schemes lie in between these
extremes \cite{tishby2015deep}, 
translating theoretical limits into engineering designs.
Today, the enormous data volumes, high-dimensional sensing modalities, and unprecedented computational capabilities, push forward the role of data-driven and model-based methods in communication system design~\cite{qiao2025token,o2017deep}. 
Specifically, deep neural networks (DNNs) 
effectively enable modeling the complex, non-linear nature of real-world data, learning compact and lower-dimensional representations from raw information sources, preserving application-relevant structure and acting as \textit{semantic} extractors~\cite{xie2021deep}.  
In this context, communication serves as a meaning exchange, 
whose effectiveness is evaluated on task-relevance \cite{ma2023task}, rather than on classical bit-fidelity metrics.

This paradigm-shift is captured in Semantic and Goal-Oriented communications (SC)~\cite{strinati20216g},
whose aim is to extract the underlying structure before information transmission, saving network resources, and enabling effective task execution~\cite{gunduz2022beyond,strinati2024goal}. However, introducing semantics
fundamentally alters the mere engineering 
problem: transmitted symbols are represented by latent vectors, whose geometry is model-and-data dependent~\cite{lobashev2025hessian}, and significance must be assessed jointly with the underlying physical layer. As a result, the intrinsic relationship between learned representations, Artificial Intelligence (AI)-native applications requirements, and physical-layer restrictions (e.g. fading, interference and latency) remains insufficiently understood.
\noindent Moreover, when semantic-aware users employ heterogeneous internal representations, their model-dependence induces misaligned latent spaces~\cite{kornblith2019similarity}, arising 
\textit{semantic noise}, i.e., errors caused by mismatched logic and background knowledge. This scenario naturally arises in decentralized semantic ecosystems composed of heterogeneous AI-native devices that do not share a common internal logic or semantic representation space. For instance, inter-operating devices produced by different vendors or manufacturers may rely on proprietary foundation models and latent representations that cannot be disclosed due to intellectual property, privacy, or commercial constraints. As a consequence, transmitter and receiver semantic spaces may exhibit substantial mismatch, despite targeting the same downstream task. A similar situation emerges in neuromorphic and ultra-low-power edge AI platforms, such as \textit{NeuroCorgi}~\cite{panades2024772muj}, where compact hardware-oriented neural architectures are optimized for local inference under stringent energy and memory constraints. In these settings, semantic communication nodes may employ highly heterogeneous latent representations, making semantic interoperability particularly challenging and motivating adaptive latent-space alignment mechanisms.

\noindent \textbf{Related works.} Semantic channel equalization (SCE) addresses latent-space mismatch by aligning encoder and decoder representations when end-to-end joint training is impractical~\cite{pandolfo2025latent}. Existing approaches exploit optimal transport~\cite{alvarez2019towards}, isometry-invariant transformations~\cite{moschella2022relative}, Parseval frames~\cite{fiorellino2025frame}, and sheaf-theoretic methods that jointly learn communication topologies and alignment maps~\cite{grimaldi2025learning}. While effective, these methods mainly focus on point-to-point semantic communication settings. However, practical semantic wireless systems are inherently multi-user and interference-limited, with multiple semantic-aware agents concurrently sharing network resources over interference channels without explicit coordination, as in classical and cognitive radio networks~\cite{merluzzi2025goal}. In such scenarios, semantic alignment must coexist with interference management and opportunistic spectrum access. To this end, recent works on semantic-aware source coding have investigated joint source--channel coding~\cite{bourtsoulatze2019deep} and goal-oriented compression strategies~\cite{shao2021learning,wang2024goal,di2023goal}, following the ``learning to communicate'' paradigm~\cite{o2017introduction}. Nevertheless, most existing approaches still neglect multi-agent interactions over interference channels.\\
By contrast, interference management in semantic-agnostic wireless systems has been extensively studied through bio-inspired resource allocation and self-organizing mechanisms~\cite{dressler2010survey,di2013bio}, as well as game-theoretic formulations for distributed power control and spectrum sharing~\cite{scutari2010convex,charilas2010survey}. In these frameworks, cognitive radio nodes compete to optimize local utility functions while exploiting sensing capabilities to mitigate, or completely avoid, interference toward primary users~\cite{scutari2008competitive,pang2010design}. Motivated by these observations, we extend semantic communication toward cognitive multi-agent systems by introducing \textit{semantic-aware cognitive users}, i.e., agents capable of processing semantic information through heterogeneous latent representations while sensing the electromagnetic environment and opportunistically accessing shared wireless resources. This perspective calls for a unified framework jointly accounting for semantic alignment, interference management, and cognitive spectrum access. Although~\cite{souza2026low} studies delay--accuracy trade-offs for goal-oriented cognitive communications, semantic mismatch and multi-user interference are not considered. To the best of our knowledge, a distributed game-theoretic framework for semantic channel equalization over cognitive interference networks remain unexplored.\\
\textbf{Contributions.}
This work investigates distributed semantic channel equalization in cognitive multi-user MIMO networks, where secondary semantic-aware agents operate concurrently without explicit coordination while preserving licensed users through spatial interference constraints. The problem is formulated as the joint optimization of semantic linear transceivers under power and interference constraints, while accounting for mutual interference among cognitive users. We model the resulting interaction as a Nash Equilibrium Problem (NEP), where each player selfishly optimizes its semantic alignment strategy through linear pre(post)-equalization transformations that jointly perform semantic compression (decompression), interference mitigation, and latent-space alignment. The main contributions of this work are summarized as follows: \vspace{-.1cm}
\begin{itemize}
\item[\textit{i)}] We introduce a distributed game-theoretic framework for semantic channel equalization over cognitive MIMO interference networks, enabling latent-space alignment among heterogeneous semantic agents operating under mutual interference and hierarchical spectrum-access constraints. The proposed formulation jointly accounts for semantic mismatch, physical-layer interactions, and opportunistic spectrum sharing within a unified framework.
\item[\textit{ii)}] We formulate the joint design of semantic linear MIMO transceivers as a constrained non-cooperative optimization problem and derive a closed-form solution for the players’ best-response strategy. Specifically, we show that the original matrix-valued optimization problem can be equivalently recast as a lower-dimensional power-allocation game over transmit antennas and channel uses.
\item[\textit{iii)}] We prove that the resulting optimal strategy admits a semantic water-filling interpretation, where the allocated power jointly depends on semantic alignment and interference conditions. Building on this structure, we propose the Iterative Semantic Water-Filling (ISWF) algorithm, enabling fully distributed latent-space alignment across cognitive semantic users without centralized coordination.
\item[\textit{iv)}] We establish sufficient conditions for existence, uniqueness, and global convergence to a pure-strategy Nash equilibrium. We also derive an explicit relationship between semantic alignment properties and the cross-coupling interactions among users, proving that the convergence of the semantic game is inherently governed by both the semantic geometry of the latent representations and the physical characteristics of the wireless channels.
\item[\textit{v)}] We assess the proposed framework on image classification and reconstruction tasks using heterogeneous semantic models, demonstrating the effectiveness of distributed semantic alignment under interference-limited conditions. Numerical results highlight the trade-offs among semantic compression, task performance, interference mitigation, and cognitive spectrum access, while validating the proposed low-complexity distributed implementation.
\end{itemize}
\textbf{Outline.} The remainder of this paper is organized as follows. Section~\ref{sec: system_model} introduces the considered cognitive semantic communication model, and Section~\ref{sec:Problem formulation} formulates the SCE problem under power and interference constraints. Section~\ref{sec:semantic games} recast the problem within a game-theoretic framework and derives an equivalent lower-dimensional semantic power-allocation game. Section~\ref{sec:ISWF} presents the proposed Iterative Semantic Water-Filling algorithm, analyzing its convergence properties and establishing sufficient conditions for uniqueness of the Nash equilibrium. Finally, Section~\ref{sec:Numerical_results} provides numerical results, and Section~\ref{sec:conclusions} draws the conclusions.\\
\textbf{Notation.} Scalar, column vector, and matrix variables are respectively indicated by plain letters $a$ (A), bold lowercase letters $\mathbf{a}$, and bold uppercase letters $\mathbf{A}$. 
The $n$-th component of a vector is indicated by $[a]_n$. 
We will refer to sets with calligraphic uppercase letters $\mathcal{A}$. The range space and the null space are denoted respectively by $\mathcal{R}(\cdot)$ and $\mathcal{N}(\cdot)$, the $n$-th eigenvalue of a matrix $\A$ is denoted by $\text{eig}_n(\A)$. 
Additional notation is introduced as needed throughout the manuscript.

\section{System Model}
\label{sec: system_model}
We consider a multi-user interference environment composed by a set of $L$ semantic-aware transmitter-receiver (tx-rx) pairs, indexed by $l$, acting as secondary users and sharing the same set of physical resources as time, frequency and space, and a set $\mathcal{P}$ of primary users owning these resources. Each secondary tx-rx pair is endowed with pre-trained DNNs to encode and decode semantic information, respectively, with the common goal to execute a down-stream task at receiver side. The described communication system coexist without direct cooperation, and no centralized master node or authority is assumed to handle the network access for secondary users. Let $\mathcal{D}$ be a shared dataset accessible to all tx-rx pairs, and $\s_{T_{l}} \!\in\! \mathbb{R}^{d_l}$ denote the semantic feature vector extracted by the $l$-th transmitter from a data point $\z \!\in\! \mathbb{R}^q$, via a DNN backbone pre-trained function.
\begin{figure*}[t]
    \centering
    \includegraphics[
        width=0.75\textwidth,
        height=0.38\textheight,
        keepaspectratio,
        trim=0.55cm 3.75cm 0.35cm 3.45cm,
        clip
    ]{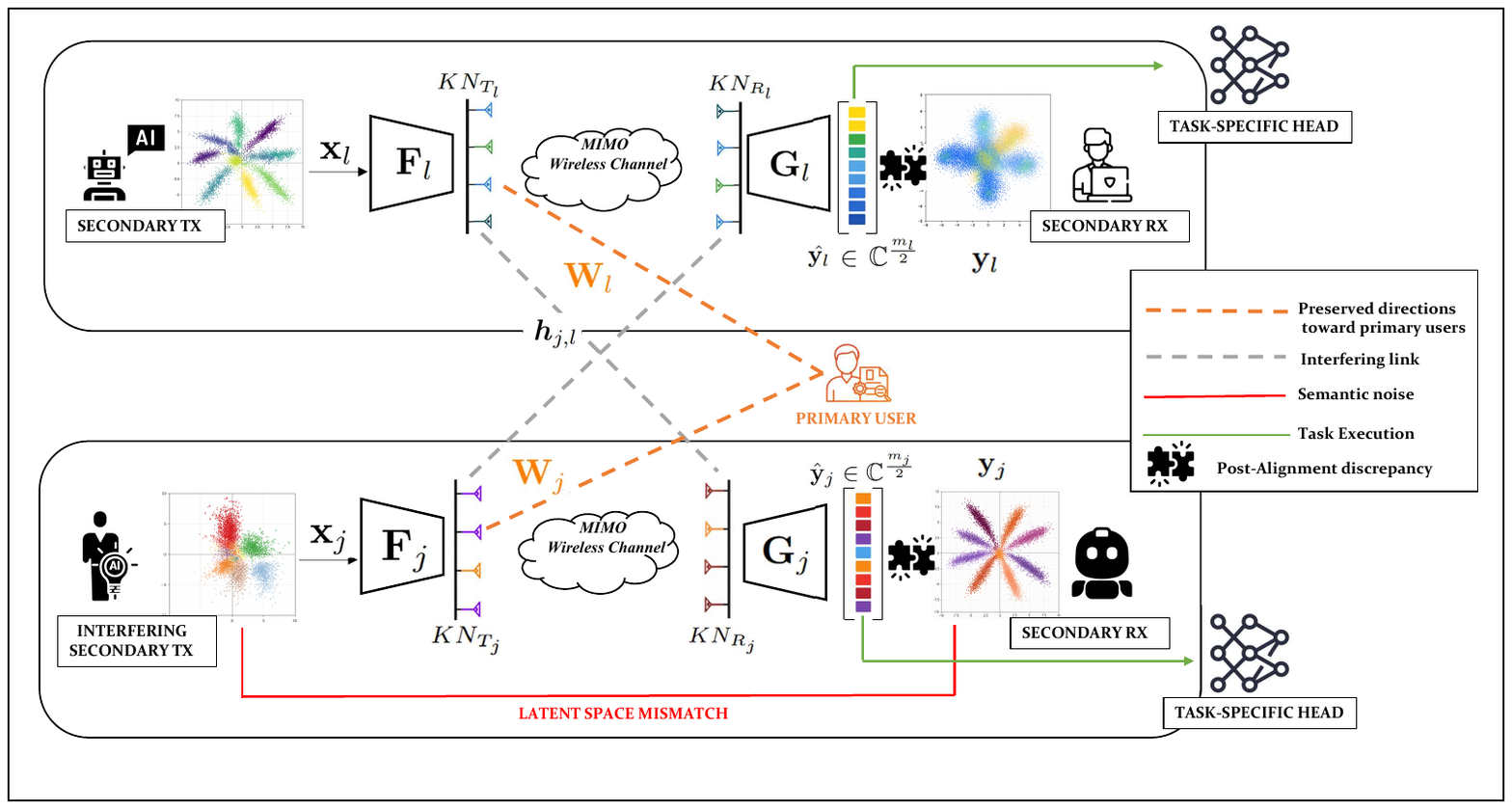}
    \caption{\small Pictorial overview of the proposed system model.}
    \label{fig:System_Model}
\end{figure*}
The collection of latent vectors $\s_{T_l}$ generated from all $\z \in \mathcal{D}$ defines the transmitter semantic latent space, which captures the internal representation used to map raw input data into lower-dimensional task-oriented features. Each $l$-th receiver relies on its own latent space structure, which differs from that of its intended transmitter, and must therefore be properly aligned to ensure reliable task performance. Furthermore, simultaneous transmissions introduce multi-user interference (MUI), which degrades the quality of the transmitted latent information. The objective is to maximize the alignment between the latent spaces of each tx-rx secondary user pair, while accounting for both (semantic and physical) channel noise and MUI, without generating interference to primary users owning the spectrum.\\
The proposed semantic equalization scheme at the $l$-th transmitter is composed of the following steps. First, assuming w.l.o.g. that $d_l$ is even, we proceed by pairing the first half of the semantic features in $\mathbf{s}_{T_l}\!\in \!\mathbb{R}^{d_l}$ with the second half to form complex symbols, yielding an input vector $\x_l \!\in\! \mathbb{C}^\frac{d_l}{2}$. Similarly to analog deep joint source-channel coding (DJSCC) schemes~\cite{bourtsoulatze2019deep,gunduz2022beyond}, latent representations are transmitted directly over the wireless channel without conventional digital modulation and channel coding, enabling end-to-end semantic-aware communications. Then, we exploit a semantic pre-equalizer $f_l(\cdot)$ that jointly performs semantic alignment and feature compression. Specifically, $f_l\!:\! \mathbb{C}^{\frac{d_l}{2}} \!\xrightarrow{}\! \mathbb{C}^{KN_{T_l}}$ performs a learnable transformation that maps the transmitter complex vector $\x_l\!\in\!\mathbb{C}^\frac{d_l}{2}$ into the compressed representation $\overline{\mathbf{x}}_l\in\mathbb{C}^{K N_{T_l}}$, where $N_{T_l}$ is the number of antennas at the $l$-th transmitter, and $K$ represents the number of channel uses. The compression factor resulting from transmitting $K$ MIMO symbols from the original $d_l/2$ complex values is given by
$\xi_l \!=\! \frac{K}{d_l/2}$. The compressed vectors $\overline{\mathbf{x}}_l$ are transmitted over $K$ channel uses through a flat-fading MIMO channel $\overline{\mathbf{H}}_{l,l} \in \mathbb{C}^{N_{R_l} \times N_{T_l}}$, where $N_{R_l}$ is the number of antennas at the $l$-th receiver. 
In addition, concurrent transmission causes interference to the other users through the cross-link MIMO channels, i.e., $\mathbf{H}_{j,l}$, for all $j \!\neq\! l$. Eventually, at the rx-side, a post-equalization function maps the received symbols into a complex vector $\hat{\y_l} \!\in\! \mathbb{C}^{\frac{m_l}{2}}$ via learnable transformation $ g_l \!:\! \mathbb{C}^{KN_{R_l}} \!\xrightarrow{}\! \mathbb{C}^{\frac{m_l}{2}}$.  Overall, the received signal over the $l$-th tx-rx pair can be modeled as:
\vspace{-0.15 cm}
\begin{equation}
  \hat{\mathbf y}_l
  = g_l\!\left(
      \mathbf H_{l,l} f_l(\mathbf x_l)
      + \sum\nolimits_{j \ne l} \mathbf H_{j,l} f_j(\mathbf x_j)
      + \mathbf v_l
    \right)
    \label{eq:received_signal_vectorized}
\end{equation} 
where $\H_{l,l} \!=\! \I_K  \otimes \overline{\H}_{l,l} \in \mathbb{C}^{KN_{R_l} \times KN_{T_l}}$ is the direct channel of link $l$, $\H_{j,l} = \I_K  \otimes \overline{\H}_{j,l}  \in \mathbb{C}^{KN_{R_l} \times KN_{T_j}}$ is the cross-channel matrix between source $j$ and destination $l$, with $\otimes$ denoting the Kronecker product; finally, $\v_l$ is a zero-mean circularly symmetric complex Gaussian noise vector with covariance matrix $\R_{v_l}\!=\!\sigma_{v_l}^2 \I$.
The second term on the right-hand side of $(\ref{eq:received_signal_vectorized})$ represents the MUI perceived by the $l$-th destination and caused by the other active communication links. To ease notation, in the sequel we denote the MUI plus noise (MUIN) term of the $l$-th link in $(\ref{eq:received_signal_vectorized})$ as $$\n_l= \sum\nolimits_{j \ne l} \mathbf H_{j,l} f_j(\mathbf x_j) + \mathbf v_l.$$
At the $l$-th receiver, the complex latent vector  $\hat{\y_l} \in \mathbb{C}^{\frac{m_l}{2}}$ in $(\ref{eq:received_signal_vectorized})$ is then converted into a real vector $\hat{\s}_{R_l}$ of dimension $m_l$, by inverting the halving operation done at transmitter side. Finally, the received signal is processed by a task-specific DNN at the receiver to accomplish a desired task.\\
\noindent\textbf{Semantic channel equalization model.} The semantic mismatch problem can be addressed by minimizing a post-alignment discrepancy metric $d_0(\cdot,\cdot)$ between the latent representation expected at the $l$-th receiver, $\y_l$, and the aligned received signal $\hat{\mathbf y}_l$, namely
\vspace{-0.15 cm}
\begin{equation}\label{eq:d0}
    \min_{\hat{\mathbf y}_l \in \mathcal{S}}
    d_0(\y_l,\hat{\mathbf y}_l)
\end{equation}
subject to the physical communication constraints encoded in the feasible set $\mathcal{S}$. 
The goal is to align the semantic representations exchanged among AI-native devices while respecting the underlying communication constraints. Several alignment metrics can be employed for this purpose, ranging from optimal transport mappings to information-theoretic dissimilarity measures between distributions, such as the Kullback--Leibler divergence or the Hellinger distance \cite{polyanskiy2025information}. In this work, we jointly optimize the learnable transformations $f_l$ and $g_l$ for each unlicensed user $l$, aiming to minimize the semantic discrepancy between the target latent vectors $\mathbf{s}_{R_l}$ and the reconstructed representations $\hat{\mathbf{s}}_{R_l}$ across all $L$ communication links.
To enable a tractable analytical formulation, we consider linear semantic pre-equalizers $\{f_l(\cdot)\}_{l=1}^L$, each represented by a matrix $\mathbf{F}_l \!\in\! \mathbb{C}^{K N_{T_l} \times \frac{d_l}{2}}$, and linear semantic equalizers $\{g_l(\cdot)\}_{l=1}^L$, modeled by matrices $\mathbf{G}_l \!\in\! \mathbb{C}^{\frac{m_l}{2} \times K N_{R_l}}$, for all $l \!=\! 1, \ldots, L$.
To this end, we exploit latent training vectors as \textit{semantic pilots} to enable semantic channel estimation and latent-space alignment~\cite{pandolfo2025latent}. Specifically, semantic pilots consist of paired latent representations extracted from semantically related data samples at transmitter and receiver side, and are used to learn the transformations required to reconcile heterogeneous semantic spaces through the wireless channel. For the $l$-th link, we consider a set of $n$ labeled pilot pairs
$\{(\mathbf{x}_{i,l},\mathbf{y}_{i,l})\}_{i\in\mathcal T_r},$ drawn from the available training dataset $\mathcal T_r$, where $\mathbf{x}_{i,l}$ and $\mathbf{y}_{i,l}$ denote the tx and rx latent representations associated with the same semantic sample.
We make the following assumptions:
\begin{itemize}
    \item[A0)] \textit{Pre-whitening:} latent space vectors $\x_l,\y_l$ are modeled as zero-mean random vectors with weighted covariance matrices $\R_x^{(\boldsymbol\Omega)},\R_y^{(\boldsymbol\Omega)}=\I$, for any diagonal weighting matrix $\boldsymbol\Omega\succ0$. Here, $\boldsymbol\Omega$ weights the semantic pilot samples (cf.~\eqref{eq:nonconvexERM}), emphasizing semantically relevant examples or latent directions. This assumption is without loss of generality, since the corresponding whitening transformation can always be absorbed into the semantic pre-equalizer $\F_l$.
    \label{assumption 0: pre-withening}
    \item[A1)] \textit{Statistically independent representations:} we assume
\begin{equation}
\{\x_l,\y_l\} \independent \x_j \qquad \forall\, l \neq j \in L.
\label{assumption:independence}
\end{equation}
with $\independent$ denoting independence among random variables.
     \item[A2)] \textit{
     Perfect channel state information (CSI):} we assume closed loop systems with perfect CSI,  with a sufficiently long channel coherence time such that channels are considered static across all channel uses. Channel matrices are assumed to be full rank. The effect of imperfect CSI knowledge will be numerically assessed in Sec.\ref{sec:Numerical_results}.
     \item[A3)] 
     \textit{Interference estimation:} each receiver is able to estimate the covariance matrix of the locally perceived interference generated by the other communication links.
    \label{assumption 1: stat. independence}
\end{itemize}
Under these assumptions, the channel model over the $l$-th communication link $(\ref{eq:received_signal_vectorized})$ boils down to:     
\begin{equation}
  \hat{\y}_l
    = \G_l\H_{l,l} \F_l \x_l + \G_l\n_l,
    \quad l=1,\ldots,L.
    \label{eq:received_signal_matrices}
\end{equation}
Furthermore, for each $l$-th transmitter, the total average transmit power is constrained as:
\vspace{-0.15 cm}
\begin{equation}
 \mathbb{E}\{ ||\F_l\x_l||_F^2 \}= \text{tr} \{\F_l\F_l^H \} \leq KP_{\max},\quad \forall l \in L  
\label{eq:power_constraint}
\end{equation}
where $P_{\max}$ is the average power budget per channel use.
Assumption A1 is well motivated in decentralized semantic networks, where heterogeneous users rely on independently trained models and uncoordinated data sources. Assumption A3 further captures cognitive semantic users with sensing capabilities, enabling adaptive transmission while limiting interference toward primary users. Motivated by this setting, we next formulate the semantic alignment problem within a cognitive multi-user communication framework.

\section{Problem formulation}
\label{sec:Problem formulation}

Building on the semantic communication model introduced in the previous section, the semantic channel equalization problem can be formulated as an empirical risk minimization problem under power and interference constraints. For analytical tractability, we adopt the weighted mean squared error (MSE) as post-alignment distance metric $d_0(\cdot,\cdot)$ in (\ref{eq:d0}). Semantic linear transceivers jointly perform semantic compression (decompression), physical-layer equalization, and latent-space alignment at transmitter (receiver) side. To preserve licensed users while enabling spectrum sharing, we also impose the deterministic spatial interference constraint
\begin{equation}
    \W_l^H \F_l \F_l^H = \mathbf{0},
\label{eq:null shaping constraint}
\end{equation}
where $\W_l \in \mathbb{C}^{KN_{T_l} \times r_{W_l}}$ is a full-rank matrix whose columns span the protected spatial directions toward the primary users. Hence, secondary transmissions are forced to lie in the orthogonal complement of the protected subspace, enabling spatial spectrum sharing. The resulting optimization problem is:
\begin{align}\label{eq:nonconvexERM}
\min_{\mathbf{F}_l, \mathbf{G}_l} \;\;
& 
\sum_{i \in \mathcal{T}_r}
\mathbb{E}\left\{
\omega_{i,l}
\left\|
\mathbf{y}_{i,l}
-
\mathbf{G}_l
\left(
\mathbf{H}_{l,l} \mathbf{F}_l \mathbf{x}_{i,l}
+
\mathbf{n}_l
\right)
\right\|_F^2
\right\}
 \\
&\text{s.t.} \quad
 \mathrm{tr}\!\left\{\mathbf{F}_l \mathbf{F}_l^H\right\}
\leq K P_{\max},
\quad
\W_l^H (\F_l \F_l^H)=\mathbf{0}, \notag
\end{align}
for all $l=1,\ldots,L,$ where
$\boldsymbol\Omega_l \triangleq \operatorname{diag}(\boldsymbol\omega_l)
\in \mathbb{R}^{n\times n}$
contains strictly positive coefficients weighting the post-alignment distortion metric. Under A0 and A1, the covariance matrix of the multi-user interference-plus-noise (MUIN) term is given by
\begin{equation}
    \R_{n_l}=
\sum\nolimits_{j \neq l}
    (\H_{j,l} \F_j)
    (\H_{j,l} \F_j)^H
    +
    \sigma_{v_l}^2 \boldsymbol\R_{v_l}^{(\boldsymbol{\Omega_l})}
\label{eq:MUI_covariance}
\end{equation}
where
\begin{equation}
\R_{v_l}^{(\boldsymbol\Omega_l)}
=
\operatorname{tr}(\boldsymbol\Omega_l)\I
\in
\mathbb{C}^{KN_{R_l}\times KN_{R_l}}
\end{equation}
denotes the weighted noise covariance matrix. The covariance in \eqref{eq:MUI_covariance} captures both thermal noise and the aggregate multi-user interference generated by the other active links,
shedding light on the underlying user's interactions. Problem~\eqref{eq:nonconvexERM} is challenging for two main reasons. First, the objective function depends bilinearly on the semantic transceivers $(\F_l,\G_l)$, yielding a non-convex optimization problem. Second, communication links are mutually coupled through the interference term $\mathbf{G}_l\mathbf{n}_l$, since each user’s strategy depends on the semantic transmission policies adopted by the others. These interactions naturally induce a decentralized strategic behavior among semantic users. Motivated by this observation, in the next section we recast the SCE problem as a non-cooperative game, where each communication link selfishly optimizes its semantic transceiver under cognitive interference constraints.

\section{Semantic Games for Latent Space Alignment}
\label{sec:semantic games}

Due to the interference coupling encoded in~\eqref{eq:MUI_covariance}, the SCE problem naturally induces a decentralized strategic interaction among communication links. Motivated by this observation, we reformulate~\eqref{eq:nonconvexERM} as a Nash Equilibrium Problem (NEP), modeling each communication link as a strategic player whose optimization variables depend on the strategies adopted by the other active users. Formally, the $L$ communication links act as players, where each $l$-th user controls its alignment strategy $\boldsymbol{\theta}_l\triangleq[\G_l,\F_l]$, with the semantic pre-equalizer constrained to the local feasible set $\mathcal{Q}_l$, defined as:
\begin{equation}
\mathcal{Q}_l \!\triangleq\!\!
\left\{
\F_l \!\in\! \mathbb{C}^{ K N_{T_l} \times \frac{d_l}{2}} \!:\!
\operatorname{tr}\!\left\{\F_l \F_l^H \right\} \!\le\! K P_{\max},\!
\W_l^H \F_l \F_l^H\!=\!0
\right\}\!.
\label{eq:local feasible set}
\end{equation}
Since \!no \!coupling constraints are present among players, the global feasible \!set \!is \!the \!Cartesian product \!$\mathcal{Q} \!\triangleq\! 
\!\Pi_{l=1}^L
\mathcal{Q}_l
$. Accordingly, the aggregated decision variable is denoted by \!$\boldsymbol{\theta} \!\triangleq\! (\boldsymbol{\theta}_1, \!\dots\!, \boldsymbol{\theta}_L) \!\in\! \mathcal{Q}$.
Each $l$-th player aims to choose a profile strategy that maximizes its payoff function, denoted by $p_l(\cdot)$:
\begin{equation}
\begin{aligned}
\max_{\boldsymbol{\theta}_l \in \mathcal{Q}_l}\;\; & p_l\!\left(\boldsymbol{\theta}_l,\boldsymbol{\theta}_{-l}\right), 
\end{aligned}
\label{eq:NE-problem}
\end{equation}
where \!$\boldsymbol{\theta}_{-l} \!\triangleq\! \bigl(\boldsymbol{\theta}_{1},\, \!\ldots\!,\, \boldsymbol{\theta}_{l-1},\, \boldsymbol{\theta}_{l+1},\, \!\ldots\!,\, \boldsymbol{\theta}_{L}\bigr)$ \!identify the alignment strategies of interfering links. 

An aggregate strategy profile $\boldsymbol{\theta}^{\star}$ is said to be a NE if no player can improve its objective by unilaterally deviating from its own equilibrium strategy $\boldsymbol{\theta}_l^{\star}$,
given that all the other players act according to it \cite{scutari2010convex}. 
Formally, it can be expressed as: 
\begin{equation}
p_l(\boldsymbol{\theta}_l^{\star}, \boldsymbol{\theta}_{-l}^{\star}) \geq p_l(\boldsymbol{\theta}_l, \boldsymbol{\theta}_{-l}^{\star}),  \quad \forall \boldsymbol{\theta}_l \in \mathcal{Q}_l.
    \label{eq: NE_condition}
\end{equation}
Generally, the achievability and convergence to a NE holds only under some conditions. Following \textit{Rosen's theorem} \cite{rosen1965existence}, a game $\mathcal{G}: \langle L,\mathcal{Q}, \{p_l\}_{l \in L} \rangle$ admits at least one (pure) NE if:
\begin{itemize}
    \item[(i)] \textit{for every player, the payoff function is continuously differentiable and concave in $(\boldsymbol{\theta}_l, \boldsymbol{\theta}_{-l})$ given the strategies of other players $\boldsymbol{\theta}_{-l}$ };
    \item[(ii)] \textit{each $l$-th nonempty feasible set $\mathcal{Q}_l$ is compact and convex.}
\end{itemize}

\noindent 
In the considered setting, the objective function in~\eqref{eq:nonconvexERM} naturally defines the local utility of the $l$-th player. However, its non-convex structure hinders the analysis of the resulting NEP. Therefore, in the sequel, we derive an equivalent convex reformulation enabling tractable optimization of semantic linear MIMO transceivers, and guaranteeing the existence of a NE.

\subsection{Pre-Equalizer only Formulation}
\noindent We start observing that, for a fixed $\mathbf{F}_l$, problem~\eqref{eq:nonconvexERM} becomes convex in $\mathbf{G}_l$. Let $\X_l\in\mathbb{C}^{\frac{d_l}{2}\times n}$ denote the matrix that collects all latent transmitter samples $\{\mathbf{x}_{i,l}\}_{i\in\mathcal{T}_r}$, and $\Y_l\in\mathbb{C}^{\frac{m_l}{2}\times n}$ be the matrix containing the corresponding receiver latent column vectors $\{\mathbf{y}_{i,l}\}_{i\in\mathcal{T}_r}$. Then, the optimal equalizer $\mathbf{G}_l$ admits the closed-form solution
\begin{align}
    \G_l^{opt} =
    \Y_l \boldsymbol{\Omega_l} \X_l ^H 
    (\H_{l,l}\F_l)^H  (\H_{l,l}\F_l 
    \F_l^H \H_{l,l}^H + 
    \boldsymbol \R_{n_l})^{-1}
    \label{eq:Wiener_filter}
\end{align}
that is optimal at the $l$-th receiver for any given semantic pre-equalizer $\mathbf{F}_l$. As derived in Appendix~\ref{appA}, exploiting~\eqref{eq:Wiener_filter} and defining the weighted semantic cross-covariance matrix
$\P_l=\Y_l\boldsymbol{\Omega}_l\X_l^H$,
the $l$-th objective term in~\eqref{eq:nonconvexERM} admits the equivalent pre-equalizer-only formulation
\begin{equation}
    \text{MSE}_l(\F_l)\! = \! 
    \text{tr} \{ \R_{y_l}^{(\boldsymbol\Omega_l)}\! 
    -\!  \P_l \P_l^H \!+\!  \P_l(\F_l^H \R_{H_l} \F_l \!+ \!\I)^{-1}\P_l^H\},
\label{eq:objective_MSE(F)_reformulated}
\end{equation}
where $\R_{y_l}^{(\boldsymbol\Omega_l)}$ denotes the $\boldsymbol\Omega_l$-weighted covariance matrix of the receiver latent representations, and
\begin{equation}
\R_{H_l}=\R_{H_l}(\F_{-l})
=
\H_{l,l}^H
\R_{n_l}^{-1}
\H_{l,l}
\in
\mathbb{C}^{KN_{T_l}\times KN_{T_l}}
\end{equation}
is the effective channel covariance matrix induced by the MUIN term. Since $\R_{n_l}$ depends on the interfering pre-equalizers $\F_{-l}$, the matrix $\R_{H_l}(\F_{-l})$ captures the strategic coupling among semantic users.
In the sequel, we omit the explicit dependence on $\F_{-l}$ to ease the notation. Thus, hinging on  $(\ref{eq:objective_MSE(F)_reformulated})$ and omitting constant terms, we can recast $(\ref{eq:nonconvexERM})$ as:
\begin{equation}
\label{eq:non-convex-problem_wrt_F}
\begin{aligned}
\min_{\F_l} \;\;
& 
\operatorname{tr}\!\left\{\left(\F_l^{H} \R_{H_l} \F_l + \I \right)^{-1} \P_l^H \P_l \right\} \\
&\text{s.t.}\quad
 \operatorname{tr}\!\left\{\F_l \F_l^H \right\} \le K P_{\max}, \quad \W_l^H( \F_l \F_l^H)=\mathbf{0} .
\end{aligned}
\end{equation}
Although \eqref{eq:non-convex-problem_wrt_F} depends only on the semantic pre-equalizer $\F_l$, the resulting optimization problem remains non-convex. In particular, the null-interference constraint couples the transmit covariance matrix with the preserved secondary-to-primary spatial directions, making a direct convex reformulation with respect to $\F_l$ challenging. To overcome this issue, in the following we derive an equivalent lower-dimensional formulation that explicitly incorporates the interference constraints into the feasible transmit subspace.

\subsection{Pre-Equalizer design based on Subspace Projections}\label{sec: subspace projections}
Let us introduce the projection operator onto the orthogonal complement of the steering vector matrix $\W_l$, defined as:
\begin{equation}
    \boldsymbol\Pi_{\mathcal{R}(\W_l)^{\perp}} = \I - \W_l (\W_l^H \W_l)^{-1} \W_l^H,
    \label{eq: projector onto W_perp space}
\end{equation}
where $\!\mathcal{R}(\W_l)\!^{\perp}\!=\!\mathcal{N}(\W_l^H)$ specifies the orthogonal complement of the range space of $\W_l$, and $\boldsymbol\Pi_{\mathcal{R}(\W_l)^{\perp}} \!\in\! \mathbb{C}^{KN_{T_l} \times KN_{T_l}}$.
Exploiting \eqref{eq: projector onto W_perp space}, the null-interference constraint in \eqref{eq:non-convex-problem_wrt_F} can be equivalently embedded into the tx covariance structure as
\begin{equation}
    \F_l\F_l^H =
    (\boldsymbol\Pi_{\mathcal{R}(\W_l)^{\perp}} \F_l)
    (\F_l^H \boldsymbol\Pi_{\mathcal{R}(\W_l)^{\perp}}),
    \label{eq: equivalent null_constraint}
\end{equation}
which constrains the transmitted signal covariance to lie entirely within the subspace orthogonal to the protected spatial directions associated with the primary users. By using \eqref{eq: equivalent null_constraint} into the objective of \eqref{eq:non-convex-problem_wrt_F}, and replacing the cross-channel matrices with their projected counterparts $\H_{j,l}^{\perp} = \H_{j,l}\boldsymbol\Pi_{\mathcal{R}(\W_j)^{\perp}}$, the local problem of the $l$-th secondary user admits the following equivalent optimization formulation:
\begin{equation}
\label{eq:projected_non-convex-problem_wrt_F}
\begin{aligned}
\min_{\F_l} \;\;
& 
\operatorname{tr}\!\left\{\left(\F_l^{H} \R_{H_l}^{\perp} \ 
\F_l + \I \right)^{-1} \P_l^H \P_l \right\} \\
&\text{s.t.}\quad
 \operatorname{tr}\!\left\{\F_l \F_l^H \right\} \le K P_{\max},  
\end{aligned}
\end{equation}
where the channel covariance $\R_{H_l}^{\perp} 
\!\in\!\mathbb{C}^{K N_{T_l} \!\times\! K N_{T_l}}$ reads as:
\begin{equation}
\R_{H_l}^{\perp} \!\triangleq\! 
\boldsymbol\Pi_{\mathcal{R}(\W_l)^{\perp}}
( \mathbf{H}_{l,l}^H \mathbf{R}_{{n}_l}^{-1} \mathbf{H}_{l,l}
) 
\boldsymbol\Pi_{\mathcal{R}(\W_l)^{\perp}}.
\label{eq:projected channel covariance}
\end{equation}
Defining the (possibly rank-deficient) matrix 
$\H_{ll}^{\perp}\!\triangleq\! \H_{ll}\boldsymbol\Pi_{\mathcal R(\W_l)^\perp}\!\in\!\mathbb{C}^{KN_{R_l}\!\times\!KN_{T_l}}$, with rank 
$r_{\H_l^\perp}$,
any optimal transmission strategy solving \eqref{eq:projected_non-convex-problem_wrt_F} belongs to the orthogonal complement of the null space of $\H_{ll}^{\perp}$, regardless of the MUIN covariance structure. Accordingly, let
$\H_{l,l}^{\perp}
=
\M_{l,1}\boldsymbol\Psi_{ll}\D_{l,1}^H$
denote the reduced singular value decomposition of $\H_{l,l}^{\perp}$, where
$\D_{l,1}\in\mathbb{C}^{KN_{T_l}\times Kr_{\H_l^{\perp}}}$
is a semi-unitary matrix spanning the orthogonal complement of
$\mathcal{N}(\H_{l,l}^{\perp})$, and
$\boldsymbol\Psi_{ll}\succ0$ collects the non-zero singular values. Then, without loss of optimality, the semantic pre-equalizer can be parameterized as
\[
\F_l=\D_{l,1}\bar{\F}_l
\]
with reduced-dimensional variable
\begin{equation}
\bar\F_l\in\mathcal{\bar Q}_l\triangleq
\left\{
\bar\F_l\in\mathbb{C}^{Kr_{\H_l^\perp}\times d_l/2}
:
\operatorname{tr}\{\bar\F_l\bar\F_l^H\}\leq KP_{\text{max}}
\right\}.
\end{equation}
Substituting the factorization $\F_l=\D_{l,1}\bar{\F}_l$ into
\eqref{eq:projected_non-convex-problem_wrt_F}, we obtain the following equivalent reduced-dimensional problem:
\begin{equation}
\label{eq:lower dimensional non-convex wrt F}
\begin{aligned}
\min_{\bar{\F}_l \in \mathcal{\bar Q}_l}\;\;
& 
\operatorname{tr}\!\left\{
\left(\bar{\F}_l^H \bar{\R}_{H_l}\bar{\F}_l+\I\right)^{-1}
\P_l^H\P_l
\right\},
\end{aligned}
\end{equation}
where
$\bar{\R}_{H_l}
\triangleq
\D_{l,1}^H \R_{H_l}^{\perp}\D_{l,1}
=
\bar{\H}_{l,l}^H\bar\R_{n_l}^{-1}\bar{\H}_{l,l}$, and $\bar{\H}_{l,l}\triangleq \H_{ll}^{\perp}\D_{l,1}.$
By construction, $\bar{\H}_{l,l}\in
\mathbb C^{K N_{R_l}\times r_{\H_l^\perp}}$ has full column rank, which implies $\bar{\R}_{H_l}\succ0$. Similarly, defining the reduced cross-channel matrices as
$\bar\H_{j,l}\triangleq\H_{jl}^{\perp}\D_{j,1}$, the corresponding reduced MUIN covariance matrix is given by
\begin{equation}
    \bar\R_{n_l}
    \triangleq
    \sum\nolimits_{j \neq l}
    (\bar\H_{j,l} \bar\F_j)
    (\bar\H_{j,l} \bar\F_j)^H
    +
    \sigma_{v_l}^2 \operatorname{tr}( \boldsymbol\Omega_l) \I .
\label{eq:MUI_covariance_reduced}
\end{equation}
In the sequel, we denote by $\overline{KN_{T_l}}\triangleq r_{\bar\H_l}$ the rank of the reduced channel $\bar\H_{l,l}$, corresponding to the number of admissible channel eigenmodes available to the $l$-th user.\\
\noindent \textbf{Optimal solution via scalarization.} Interestingly, a closed-form solution to (\ref{eq:lower dimensional non-convex wrt F}) is attainable in the absence of semantic compression, i.e., when 
$\overline{KN_{T_l}} \!=\!d_l/2 $. However, in the more general case $\overline{KN_{T_l}} \! < \! d_l/2$, we must instead rely on an approximate formulation of (\ref{eq:lower dimensional non-convex wrt F}), which yields a closed-form expression for the optimal pre-equalizer. To this aim, let us introduce the following matrix decompositions: 
\begin{equation}
\bar\R_{H_l}\!=\!
\V_{h_l}\mathbf{\Lambda}_{h_l}\V_{h_l}^H \ , \; \quad
\tilde{\P}_l = \tilde{\U}_{p_l}\tilde{\mathbf{\Sigma}}_{p_l}\tilde{\Q}_{p_l}^H
\label{eq:channel_pilot_decomposition},
\end{equation}
where $\!\mathbf{\Lambda}_{h_l}\!\in \!\mathbb{R}^{\overline{KN_{T_l}}\times \overline{KN_{T_l}}}$, 
and $\tilde{\mathbf{P}}_l$ denotes the best rank-$\overline{KN_{T_l}}$ approximation of $\mathbf{P}_l$ in (\ref{eq:non-convex-problem_wrt_F}). Specifically, $\tilde{\U}_{p_l} \!\in\! \mathbb{C}^{\frac{m_l}{2} \times \overline{KN_{T_l}}}$ and $\tilde{\Q}_{p_l} \in \mathbb{C}^{ \frac{d_l}{2} \times \overline{KN_{T_l}}}$ in (\ref{eq:channel_pilot_decomposition}) are semi-unitary singular vector matrices, 
executing the best rank approximation
over the first $\overline{KN_{T_l}}$ orthonormal columns, encoding the principal directions of the semantic pilots cross-covariance. Clearly, if $\overline{KN_{T_l}} \!=\! d_l/2$, we have $\tilde{\mathbf{P}}_l\!=\!\P$. Now, leveraging the decompositions in (\ref{eq:channel_pilot_decomposition}), we can design the semantic pre-equalizer $\bar\F_l \in \mathbb{C}^{\overline{KN_{T_l}} \times \frac{d_l}{2}}$ solving (\ref{eq:lower dimensional non-convex wrt F}) as\footnote{This structure follows from jointly diagonalizing the effective channel covariance $\bar\R_{H_l}$ and the semantic cross-covariance $\P_l^H\P_l$, where $\V_{h_l}$ identifies the admissible channel eigenmodes and $\tilde{\Q}_{p_l}$ the principal semantic alignment directions. The diagonal matrix $\mathbf{\Phi}_l$ then allocates power according to semantic relevance and channel quality.}:
\begin{align}
    \bar\F_l = \V_{h_l} \operatorname{diag}(\sqrt{\mathbf{\Phi}_l}) \tilde{\Q}^H_{p_l} 
\label{eq:F_definition}
\end{align}
where \!$\boldsymbol{\Phi}_l \!\in\! 
\mathbb{R}^{\overline{KN_{T_l}} \!\times \overline{KN_{T_l}}}$
\!is \!a diagonal matrix that specifies the power allocation across the $l$-th transmitter's antennas\footnote{
Without any loss of generality, we consider the coefficients \!$\operatorname{diag}(\boldsymbol{\sqrt{\Phi_l}})$ as real,
since any phase shift can be incorporated in the columns of \!$\V_{h_l}$.}, i.e., $\boldsymbol{\varphi}_l \!\triangleq\! \text{diag}(\boldsymbol{\Phi}_{l})\!=\![\varphi_{l,1},\!\ldots\!,\varphi_{l,\overline{KN_{T_l}}}]^T$. Finally, exploiting (\ref{eq:F_definition}) and substituting $\P_l$ with $\tilde{\P}_l$ in (\ref{eq:channel_pilot_decomposition}), the objective of \eqref{eq:lower dimensional non-convex wrt F} can be approximated as: 
\begin{equation}
\mathrm{MSE}_l(\boldsymbol{\Phi}_l)
\!\approx\! 
\operatorname{tr}\!\left\{\!(\mathbf I + \tilde{\mathbf Q}_{p_l}\,
\mathbf{\Phi}_l
\boldsymbol{\Lambda}_{h_l}\,\tilde{\mathbf Q}_{p_l}^H)^{-1}
\big(\tilde{\mathbf Q}_{p_l}\,\tilde{\boldsymbol{\Sigma}}^2_{p_l}\,\tilde{\mathbf Q}_{p_l}^H\big)\!\right\}\!.
\label{eq:mse_Phi_diagonal}
\end{equation}
Now, invoking the matrix inversion lemma and the cyclic invariance of the trace, $(\ref{eq:mse_Phi_diagonal})$ admits a full diagonal reformulation in $\mathbf{\Phi}_l$, leading to the equivalent power allocation problem:
\vspace{-0.10 cm}
\begin{align} 
\max_{ \mathbf{\Phi}_l} & \quad 
\operatorname{tr} \Bigl\{  \Bigl[  (
\mathbf{\Phi}_l
\boldsymbol{\Lambda}_{h_l} )^{-1} + \I \Bigr]^{-1} \tilde{\mathbf{\Sigma}}^2_{p_l}\Bigr\} 
\notag \\
    \text{s.t.} & \quad \text{tr}\{ \mathbf{\Phi}_l 
    \} \leq KP_{\max}.
\label{eq:diagonal_Phi_optimization_problem}
\end{align}
Leveraging this decomposition, the $l$-th payoff function in~\eqref{eq:diagonal_Phi_optimization_problem} admits a fully diagonal pre-equalizer-only formulation. In this setting, the semantic and physical channel eigenmodes are jointly exploited to determine the optimal transmission directions over the MIMO interference channel. The resulting alignment strategy lies in the semantic subspace spanned by $\tilde\Q_{p_l}^H$ and incorporates the null-interference projection through the isometry $\V_{h_l}$, while the optimization variables reduce to the per-mode power allocation vector $\boldsymbol\varphi_l=\operatorname{diag}(\boldsymbol{\Phi}_l)$. Interestingly, the proposed reformulation allows us to recast the matrix-valued semantic channel equalization problem in~\eqref{eq:nonconvexERM} into a lower-dimensional \textit{vector power-control game}, where each player optimizes its transmit power allocation over the admissible semantic--channel eigenmodes. Specifically, defining
$\operatorname{diag}(\boldsymbol{\Lambda}_{h_l})
=
[\lambda_{l,1},\ldots,\lambda_{l,\overline{KN_{T_l}}}]^T$
and
$\operatorname{diag}(\tilde{\boldsymbol{\Sigma}}_{p_l})
=
[\tilde{\sigma}_{l,1},\ldots,\tilde{\sigma}_{l,\overline{KN_{T_l}}}]^T$,
the scalarized optimization problem associated with~\eqref{eq:diagonal_Phi_optimization_problem} reduces to the maximization of the concave payoff function
\begin{equation}
\begin{aligned}
\boldsymbol{\varphi}^{*}_l
=
\arg\max_{\boldsymbol{\varphi}_l}\;
p_l(\boldsymbol{\varphi}_l)
=
\sum_{m=1}^{\overline{KN_{T_l}}}
\frac{
(\varphi_{l,m}\lambda_{l,m})
\tilde{\sigma}^{2}_{l,m}
}{
\varphi_{l,m}\lambda_{l,m}+1
}
\quad\\
\text{s.t. }\;
\varphi_{l,m}\ge 0,\ \forall m,
\qquad
\boldsymbol{1}^T \boldsymbol{\varphi}_{l}\le K\,P_{\max},
\end{aligned}
\label{eq:power-opt}
\end{equation}
for all $l=1,\ldots,L$. The reformulation in~\eqref{eq:power-opt} reveals that semantic alignment is entirely characterized by the singular values of the semantic cross-covariance matrix $\tilde{\P}_l$, which quantify the strength of the shared latent subspaces between tx and rx representations along each admissible semantic--channel eigenmode. To exclude degenerate semantic directions and ensure that every admissible mode contributes to the alignment process, we introduce the following mild assumption:
\begin{itemize}
    \item[A4)] The semantic cross-covariance matrix $\tilde{\P}_l$ is full rank, i.e.,
    $\tilde{\sigma}_{l,m}>0$
    for
    $1\leq m\leq \overline{KN_{T_l}}$.
    \smallskip
\end{itemize}
Assumption A4 is naturally justified in decentralized semantic systems where transmitter and receiver employ independently trained latent spaces. In this setting, the semantic pilots induce statistically distinct yet semantically correlated latent representations, yielding a non-degenerate cross-covariance structure. Thus, all admissible alignment modes carry meaningful semantic information across the reduced latent subspace. Finally, under A4, the problem in~\eqref{eq:power-opt} induces the concave $N$-person game
$\mathcal{G}_{\mathrm{sem}}
\triangleq
\langle L,\mathcal{\bar Q},\{p_l\}_{l\in L}\rangle$.
Since each feasible set $\mathcal{\bar Q}_l$ is nonempty, compact, and convex, and the payoff function in~\eqref{eq:power-opt} is continuous and concave in $\boldsymbol{\varphi}_l$ for fixed $\boldsymbol{\varphi}_{-l}$, the existence of at least one pure-strategy Nash equilibrium follows directly from Rosen’s theorem for concave games~\cite{rosen1965existence}.\\
\noindent \textbf{Semantic Water-filling Power Allocation.} By Assumption A4, problem \eqref{eq:power-opt} is strictly concave and admits a closed-form solution obtained from the Karush--Kuhn--Tucker (KKT) conditions~\cite{boyd2004convex}. Specifically, the optimal power allocation across the admissible semantic-channel eigenmodes is:
\begin{equation}
  \varphi_{l,m}^{\ast}
  =
  \left[
  \frac{
  \tilde{\sigma}_{l,m}
  }{
  \sqrt{\mu_l\lambda_{l,m}}
  }
  -
  \frac{1}{\lambda_{l,m}}
  \right]_+,
  \qquad
  1\leq m\leq \overline{KN_{T_l}},
\label{eq:closed form solution phi}
\end{equation}
for all $l=1,\ldots,L$, where $[x]_+=\max(0,x)$, and $\mu_l\in\mathbb{R}_+$ denotes the Lagrange multiplier associated with the power constraint. The value of $\mu_l$ can be efficiently computed via bisection to satisfy
$\boldsymbol{1}^T \boldsymbol{\varphi}_{l}\le K\,P_{\max}$. The solution in~\eqref{eq:closed form solution phi} reveals a \textit{semantic water-filling} structure, where each player allocates power according to both the physical-channel gains $\lambda_{l,m}$ and the semantic relevance coefficients $\tilde{\sigma}_{l,m}$. In particular, each user only requires a local estimate of the MUIN covariance matrix in~\eqref{eq:MUI_covariance_reduced} to compute its optimal strategy. Specifically, the semantic game $\mathcal{G}_{\mathrm{sem}}$ admits a water-filling interpretation~\cite{palomar2005practical}, with water level $1/\sqrt{\mu_l}$ and mode-dependent weights
$\mathbf{s}_l \triangleq \tilde{\boldsymbol\sigma}_l/\sqrt{\boldsymbol\lambda_l}\in\mathbb{R}^{\overline{KN_{T_l}}}$.
Defining the normalized allocation
$\tilde{\boldsymbol\varphi}_l=\boldsymbol\varphi_l/\mathbf{s}_l$,
the power constraint becomes
$\sum_{m=1}^{\overline{KN_{T_l}}}\tilde{\varphi}_{l,m}s_{l,m}\le KP_{\max}$,
where the coefficients $s_{l,m}$ can be interpreted as the width of the water allocated over each semantic--channel eigenmode, as illustrated in Fig.~\ref{fig:water-filling illustration}. Interestingly, these weights jointly depend on the semantic alignment structure through $\tilde{\boldsymbol\sigma}_l$ and on the perceived interference through $\boldsymbol\lambda_l$, explicitly coupling the semantic and physical communication layers.

\begin{figure}[!t]
    \centering
\includegraphics[width=0.99\linewidth]{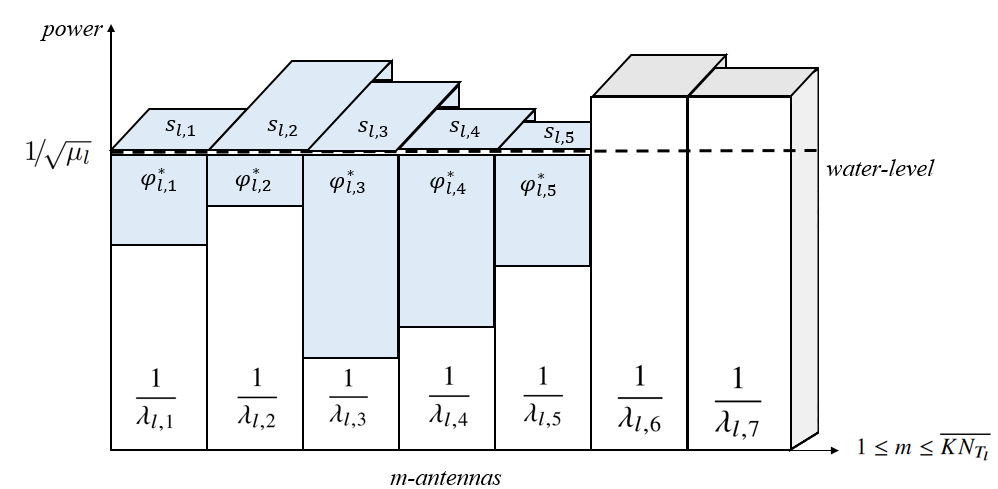}
\caption{Illustration of \textit{Semantic Water-Filling}. Bin heights represent $\boldsymbol\lambda_l^{-1}$; the water level is given by $\tfrac{1}{\sqrt{\mu_l}}$, satisfying the total average power constraint $KP_{\max}$ and weighted by the environment depended water-width $\s_l(\boldsymbol\varphi_{-l})$. The blue area shows the allocated power $\boldsymbol{\varphi}_l^\star$ across $m$-th transmitter antennas, over the space-time domain.}
    \label{fig:water-filling illustration}
\end{figure}

\section{Variational Inequality Reformulation}
\label{sec: VI reformulation}

The semantic game
$\mathcal{G}_{\mathrm{sem}}
\triangleq
\langle L,\mathcal{\bar Q},\{p_l\}_{l\in L}\rangle$
admits an equivalent reformulation as the partitioned variational inequality problem $\mathbf{VI}(\mathcal{\bar Q},\mathbf{T})$~\cite{scutari2010convex}. This reformulation is particularly useful because it enables the analysis of equilibrium existence, uniqueness, and convergence of distributed best-response dynamics through the structural properties of the pseudo-gradient mapping $\mathbf{T}$. Specifically, the goal is to find
$\boldsymbol{\varphi}^{\star}\in\mathcal{\bar Q}$
such that
\[
(\boldsymbol{\varphi}-\boldsymbol{\varphi}^{\star})^T
\mathbf{T}(\boldsymbol{\varphi}^{\star})
\geq 0,
\qquad
\forall\,\boldsymbol{\varphi}\in\mathcal{\bar Q},
\]
which directly follows from the first-order optimality conditions of the optimization problems~\cite{scutari2010convex}. The pseudo-gradient operator
$\mathbf{T}:\mathcal{\bar Q}\to\mathbb{R}^{L\overline{KN_{T_l}}}$
is defined block-wise as
\begin{equation}\label{eq:pseudogradient}
\mathbf{T}(\boldsymbol{\varphi})
\triangleq
\bigl(
-\nabla_{\boldsymbol{\varphi}_l}
p_l(\boldsymbol{\varphi})
\bigr)_{l=1}^L,    
\end{equation}
whose $m$-th component is given by
\begin{equation}
[\T_l(\boldsymbol\varphi)]_m
=
\frac{
\tilde{\sigma}_{l,m}^{2}\lambda_{l,m}(\boldsymbol\varphi_{-l})
}{
\left(
1+\varphi_{l,m}\lambda_{l,m}(\boldsymbol\varphi_{-l})
\right)^2
},
\label{eq:pseudo_graident VI}
\end{equation}
for all
$1\leq m\leq\overline{KN_{T_l}}$. In particular, strong monotonicity, or more generally the uniformly $P$ property, plays a key role in establishing uniqueness of the NE and convergence of distributed best-response dynamics~\cite{scutari2012monotone,facchinei2003finite}. Specifically, $\mathbf T$ is \emph{strongly monotone} on $\bar{\mathcal Q}$ if there exists $c_{sm}>0$ such that
\begin{equation}
(\boldsymbol{\varphi}-\mathbf z)^T
\bigl(
\mathbf T(\boldsymbol{\varphi})
-
\mathbf T(\mathbf z)
\bigr)
\ge
c_{sm}\|\boldsymbol{\varphi}-\mathbf z\|^2,
\label{eq:strong_monotonicity_definition}
\end{equation}
for all $\boldsymbol{\varphi},\mathbf z\in\bar{\mathcal Q}$.
More generally,
$\mathbf T=(\mathbf T_l)_{l=1}^L$
is a \emph{uniformly $P$-function} if there exists $c_{uP}>0$ such that
\begin{equation}
\max_{l\in L}\;
(\boldsymbol{\varphi}_l-\mathbf z_l)^T
\Bigl(
\mathbf T_l(\boldsymbol{\varphi})
-
\mathbf T_l(\mathbf z)
\Bigr)
\ge
c_{uP}\|\boldsymbol{\varphi}-\mathbf z\|^2,
\label{eq:uniformly_P_function_definition}
\end{equation}
for all $\boldsymbol{\varphi},\mathbf z\in\bar{\mathcal Q}$. Strong monotonicity implies the uniformly $P$ property, and both guarantee uniqueness of the solution of $\mathrm{VI}(\bar{\mathcal Q},\mathbf T)$~\cite{scutari2012monotone}. These properties can be characterized through the differential structure of the pseudo-gradient mapping $\mathbf T$, and, in particular, through suitable positivity and diagonal dominance conditions on its Jacobian operator~\cite{scutari2012monotone}. To this end, we characterize the Jacobian of $\mathbf T$, which quantifies how the strategy adopted by each player affects the utility gradients of the others through the semantic interference coupling:
\begin{equation}\label{eq:Jacobian}
J\mathbf T(\boldsymbol{\varphi})
\triangleq
\bigl(
J_{\boldsymbol{\varphi}_j}\mathbf T_l(\boldsymbol{\varphi})
\bigr)_{l,j=1}^L
\in
\mathbb R^{L\overline{KN_{T_l}}\times L\overline{KN_{T_l}}},
\end{equation}
whose block structure is given by
\begin{equation}
J_{\boldsymbol{\varphi}_j}\!\T_l(\boldsymbol{\varphi})\!
\!=\!
\begin{cases}
\displaystyle
\!\operatorname{diag}\!\!\left(
\frac{2\ \tilde{\sigma}_{l,m}^{\,2}\lambda_{l,m}^{2}}
{\bigl(1+\varphi_{l,m}\lambda_{l,m}\bigr)^3}
\right) \!\succ\! 0,
\!&\! j=l,\\[2ex]
\displaystyle
\!- 
\!\operatorname{diag}\!\!\left(
\tilde{\sigma}_{l,m}^{\,2}
\frac{1-\varphi_{l,m}\lambda_{l,m}} 
{\bigl(1+\varphi_{l,m}\lambda_{l,m}\bigr)^3} \right) 
\!\left[
\frac{\partial \lambda_{l,m}}{\partial \varphi_{j,n}} \right]\! 
\!&\! j\neq l.
\end{cases}
\label{eq: Jacobian structure}
\end{equation}
The diagonal blocks correspond to the Hessian of the local payoff functions and are positive definite due to strict concavity. Conversely, the off-diagonal blocks capture the interference coupling among semantic users through the sensitivity terms
$\frac{\partial\lambda_{l,m}}{\partial\varphi_{j,n}}$,
which measure how the effective channel eigenmodes vary with the rival players' strategies. To quantify these interactions, we next derive spectral and sensitivity bounds for the reduced effective channel covariance matrix.
\begin{theorem}
Let
$\boldsymbol\lambda_{l}
=
\boldsymbol\lambda_{l}(\bar\R_{H_l})
\in
\mathbb{R}^{\overline{KN_{T_l}}}$
denote the eigenvalues of $\bar\R_{H_l}$, and define
\begin{equation}\label{eq:Rnhat}
\hat{\R}_{n_l}
\triangleq
K P_{\max}
\sum\nolimits_{j \ne l}
\bar\H_{j,l}\bar\H_{j,l}^{H}
+
\sigma_{v_l}^2
\operatorname{tr}(\boldsymbol\Omega_l)\I.    
\end{equation}
Then, for all
$1\!\leq\! m\!\leq\! \overline{KN_{T_l}}$,
it holds:
\vspace{-.3 cm}
\begin{equation}
\operatorname{eig}_{l,m}
\!\left(
\bar\H_{l,l}^{H}
\hat\R_{n_l}^{-1}
\bar\H_{l,l}
\right)
\le
\lambda_{l,m}(\bar\R_{H_l})
\le
\operatorname{eig}_{l,m}
\!\left(
\frac{
\bar\H_{l,l}^{H}\bar\H_{l,l}
}{
\sigma_{v_l}^2
\operatorname{tr}(\boldsymbol\Omega_l)
}
\right).
\label{eq:R_H spectral bound}
\end{equation}
Moreover, for every $j\neq l$ and
$1\le n\le \overline{KN_{T_j}}$,
the sensitivity of
$\lambda_{l,m}(\bar\R_{H_l})$
with respect to the $j$-th player's power allocation satisfies
\begin{equation}
\left|
\frac{
\partial\lambda_{l,m}
}{
\partial\varphi_{j,n}
}
\right|
\le
\frac{
\|\bar\H_{l,l}\|_2^2
\|\bar\H_{j,l}\|_2^2
}{
\sigma_{v_l}^4
\operatorname{tr}(\boldsymbol\Omega_l)^2
}.
\label{eq: sensitivity bound Lemma}
\end{equation}
\label{lemma: R_H spectral properties}
\end{theorem}

\begin{proof}
See Appendix~\ref{app: Lemma Spectral Bound}.
\end{proof}
Theorem~\ref{lemma: R_H spectral properties} quantifies the interference coupling among semantic users through the sensitivity of the effective channel eigenmodes to the rival players' strategies. These bounds will be instrumental, in the next section, to derive sufficient conditions under which the pseudo-gradient Jacobian in~\eqref{eq:Jacobian} is diagonally dominant, thereby ensuring that the pseudo-gradient mapping in \eqref{eq:pseudogradient} is a uniformly $P$-function and, consequently, that the game admits a unique Nash equilibrium reached by distributed best-response dynamics. Such conditions identify operating regimes where the self-effect of each player dominates the interference induced by the others.

\section{Iterative Semantic Water-filling}
\label{sec:ISWF}

The closed-form solution in~\eqref{eq:closed form solution phi} naturally induces a decentralized best-response dynamics, where each semantic user updates its power allocation according to the interference generated by the other active links. Since each player only requires a local estimate of the MUIN covariance matrix in~\eqref{eq:MUI_covariance_reduced}, the resulting strategy admits a fully distributed implementation. The obtained allocation exhibits a semantic water-filling structure~\cite{palomar2005practical}, jointly balancing semantic relevance through $\tilde{\boldsymbol{\sigma}}_{l}$ and channel quality through $\boldsymbol{\lambda}_{l}$ across the admissible semantic--channel eigenmodes. In the following, we introduce the proposed Iterative Semantic Water-Filling algorithm and analyze its convergence properties.

\subsection{Distributed Algorithmic Solution}

Let $B_l(\boldsymbol{\varphi}_{-l})$ denote the best-response mapping of the $l$-th player induced by the semantic water-filling in~\eqref{eq:closed form solution phi}, namely
\begin{equation}
\label{eq:BR_definition}
B_l(\boldsymbol{\varphi}_{-l})
\triangleq
\left\{
\boldsymbol{\varphi}_l \in \mathcal{\bar Q}_l
\;\middle|\;
p_l(\boldsymbol{\varphi}_l,\boldsymbol{\varphi}_{-l})
\ge
p_l(\boldsymbol{\z}_l,\boldsymbol{\varphi}_{-l}),
\;
\forall \boldsymbol{\z}_l \in \mathcal{\bar Q}_l
\right\}.
\end{equation}
Defining the aggregate best-response operator
\[
B(\boldsymbol{\varphi})
\triangleq
B_1(\boldsymbol{\varphi}_{-1})
\times\cdots\times
B_L(\boldsymbol{\varphi}_{-L}),
\]
a strategy profile
$\boldsymbol{\varphi}^{\star}\in\mathcal{\bar Q}$
is a pure-strategy Nash equilibrium of $\mathcal{G}_{\mathrm{sem}}$
if and only if it is a fixed point of the distributed best-response dynamics, i.e.,
$\boldsymbol{\varphi}^{\star}\in B(\boldsymbol{\varphi}^{\star})$~\cite{facchinei2003finite,scutari2010convex}. Accordingly, the distributed semantic alignment dynamics can be written as the nonlinear fixed-point iteration
\begin{equation}
   \boldsymbol{\varphi}^{(t+1)}
   =
   B(\boldsymbol{\varphi}^{(t)}),
   \qquad
   t\in\mathbb{N}_{+},
\label{eq:BR dynamics}
\end{equation}
whose equilibrium coincides with the Nash equilibrium of the semantic game~\cite{combettes2021fixed}. Depending on the adopted update policy, two distributed implementations can be considered~\cite{bertsekas2015parallel}. In the Gauss--Seidel scheme, users update sequentially:
\begin{equation}
\hat{\boldsymbol{\varphi}}_{l}
=
B_{l}\!\left(
\boldsymbol{\varphi}_{1}^{\,t+1},\ldots,
\boldsymbol{\varphi}_{l-1}^{\,t+1},
\boldsymbol{\varphi}_{l+1}^{\,t},\ldots,
\boldsymbol{\varphi}_{L}^{\,t}
\right),
\label{eq:gauss_seidel_scheme}
\end{equation}
while, in the Jacobi scheme, all users update simultaneously:
\begin{equation}
\hat{\boldsymbol{\varphi}}_{l}
=
B_{l}\!\left(
\boldsymbol{\varphi}_{1}^{\,t},\ldots,
\boldsymbol{\varphi}_{l-1}^{\,t},
\boldsymbol{\varphi}_{l+1}^{\,t},\ldots,
\boldsymbol{\varphi}_{L}^{\,t}
\right).
\label{eq:jacobi_scheme}
\end{equation}
To improve the stability of the distributed dynamics, we further employ a Krasnosel'skii--Mann relaxation step~\cite{scutari2010convex}:
\begin{equation}
    \boldsymbol{\varphi}_{l}^{t+1}
    =
    \boldsymbol{\varphi}_{l}^{t}
    +
    \gamma^{t}
    \big(
    \hat{\boldsymbol{\varphi}}_{l}
    -
    \boldsymbol{\varphi}_{l}^{t}
    \big),
    \label{eq:best-response step size}
\end{equation}
with $\gamma^{t}\in(0,1]$, $\gamma^{t}\to0$, and
$\sum_{t=1}^{\infty}\gamma^{t}=+\infty$.
The resulting \textit{Iterative Semantic Water-Filling} (ISWF) algorithm is summarized in Algorithm~\ref{algo: potential game for SCE}. The proposed procedure can be interpreted as a distributed fixed-point iteration solving the nonlinear best-response dynamics in~\eqref{eq:BR dynamics}, whose convergence depends on the interference coupling among semantic users. In the next section, we derive sufficient conditions ensuring convergence of the generated sequence
$\{(\boldsymbol{\varphi}^{(t)}_l)_{l=1}^L\}_{t=1}^{T}$
toward the Nash equilibrium of $\mathcal{G}_{\mathrm{sem}}$.

\begin{algorithm}[!t]
    \caption{Iterative Semantic Water-filling}
\label{alg:zs-sce-2e}
\DontPrintSemicolon
\KwIn{
$ \F_l^{(0)}\sim \mathcal{CN}(0,1)$ s.t. $||\F_l||^2_F=KP_{\max}$} 
\KwOut{Semantic transceivers $\bar\F_l$,$\G_l$} 
\For{$t \gets 1$ in \text{game iterations}}{
\If{$\bar\F_l^{(t)}$ \text{satisfies 
termination criterion}}
{\text{STOP}}
\Else{
   Retrieve $\V_{h_l}(\boldsymbol\varphi_{-l}), \tilde\Q_{p_l}$ defined in \eqref{eq:channel_pilot_decomposition}\;
  \If{$\mathcal{B}_l\text{ scheme is Gauss-Seidel}$}{
  Compute (\ref{eq:closed form solution phi}) with the update rule (\ref{eq:gauss_seidel_scheme})}
   \ElseIf{$\mathcal{B}_l\text{ scheme is Jacobi}$}{
  Compute (\ref{eq:closed form solution phi}) with the update rule (\ref{eq:jacobi_scheme})}
}
Compute $\boldsymbol{\varphi}_l^{(t+1)}$ by (\ref{eq:best-response step size}), and
set $\bar\F_l^{(t+1)}$ as (\ref{eq:F_definition}) \;
Compute $\G_l^{(t+1)}$ by (\ref{eq:Wiener_filter}) \;
}
\Return{$\bar\F_l^{(t+1)}$,$\G_l^{(t+1)}$}\:
\label{algo: potential game for SCE}
\end{algorithm}

\subsection{Convergence Analysis}
As discussed in Section~\ref{sec: VI reformulation}, convergence of the proposed Iterative Semantic Water-Filling algorithm depends on the structural properties of the pseudo-gradient Jacobian $J\T(\boldsymbol{\varphi})$ in~\eqref{eq:Jacobian}. In particular, we seek conditions under which the self-effect of each player dominates the interference induced by the others, leading to diagonal dominance of the Jacobian operator and, consequently, to uniqueness of the Nash equilibrium and convergence of distributed best-response dynamics~\cite{scutari2010convex}. To formalize such diagonal dominance conditions, we introduce suitable lower and upper bounds on the Jacobian blocks in~\eqref{eq: Jacobian structure}. Specifically, define
\begin{equation}
\alpha_l^{\min}
\!\triangleq\!
\inf_{\boldsymbol{\varphi}\in\mathcal{Q}} \text{eig}_{\text{min}}
\!\left(
J_{\boldsymbol{\varphi}_l}\mathbf{T}_l(\boldsymbol{\varphi})
\right),
\!\quad\! 
\!\beta_{lj}^{\max}
\!\triangleq\!
\sup_{\boldsymbol{\varphi}\in\mathcal{Q}}
\left\|
J_{\boldsymbol{\varphi}_j}\mathbf{T}_l(\boldsymbol{\varphi})
\right\|_2
\label{eq:alpha_beta_def}
\end{equation}
and introduce the condensed matrix
$\boldsymbol{\Upsilon}_{\T}\in\mathbb{R}^{L\times L}$:
\begin{equation}
[\boldsymbol{\Upsilon}_{\T}]_{lj}
\triangleq
\begin{cases}
\alpha_l^{\min},
& l=j,
\\[1mm]
-\beta_{lj}^{\max},
& l\neq j.
\end{cases}
\label{eq:Upsilon_Gamma_same_row}
\end{equation}
The matrix $\boldsymbol{\Upsilon}_{\T}$ provides a compact characterization of the worst-case interaction among semantic users encoded in the Jacobian operator. In particular, $\alpha_l^{\min}$ captures the minimum self-effect of player $l$, while $\beta_{lj}^{\max}$ measures the maximum interference coupling induced by player $j$ on player $l$ through the off-diagonal Jacobian blocks. Since the local payoff function in~\eqref{eq:power-opt} is strictly concave, the diagonal terms satisfy $\alpha_l^{\min}>0$ for all $l=1,\ldots,L$. The bounds in~\eqref{eq:alpha_beta_def} follow directly from the Jacobian structure in~\eqref{eq: Jacobian structure} and the sensitivity result in Theorem~\ref{lemma: R_H spectral properties}, as illustrated in the sequel:
\begin{itemize}
\item[(i)] \textit{Diagonal blocks:}
From~\eqref{eq: Jacobian structure} and
$0\leq \varphi_{l,m}\leq KP_{\max}$,
\begin{equation}
\alpha_l^{\min}
\ge
\min_{m}\,
2\tilde{\sigma}_{l,m}^2
\frac{
\lambda_{l,m}^2
}{
(1+KP_{\max}\lambda_{l,m})^3
}
\ge
2\widetilde{\sigma}_{l,\min}^2\,\eta_l
\label{eq:alpha_bound}
\end{equation}
where
$\displaystyle\eta_l
\triangleq
\min_{m}
\frac{
\lambda_{l,m}^2
}{
(1+KP_{\max}\lambda_{l,m})^3
}.$
\smallskip
\item[(ii)] \textit{Off-diagonal blocks:}
Again from~\eqref{eq: Jacobian structure}, using sub-multiplicativity of the spectral norm and
$\left|
\frac{
1-\lambda_{l,m}\varphi_{l,m}
}{
(1+\lambda_{l,m}\varphi_{l,m})^3
}
\right|
\le1,$
yields
\begin{equation}
\beta_{lj}^{\max}
\le
\tilde{\sigma}_{l,\max}^2
\sup_{\boldsymbol\varphi\in\mathcal{\bar Q}}
\left\|
\left[
\frac{
\partial\lambda_{l,m}
}{
\partial\varphi_{j,n}
}
\right]_{m,n}
\right\|_2,
\quad
j\neq l.
\label{eq:beta_intermediate}
\end{equation}
Applying Theorem~\ref{lemma: R_H spectral properties} (cf \eqref{eq: sensitivity bound Lemma}), we finally obtain
\begin{equation}
\beta_{lj}^{\max}
\le
\tilde{\sigma}_{l,\max}^2
\frac{
\|\bar\H_{l,l}\|_2^2
\|\bar\H_{j,l}\|_2^2
}{
\sigma_{v_l}^4
\operatorname{tr}(\boldsymbol\Omega_l)^2
}.
\label{eq:beta_final_bound}
\end{equation}

\end{itemize}
Substituting the bounds (\ref{eq:alpha_bound})-(\ref{eq:beta_final_bound}) into~\eqref{eq:Upsilon_Gamma_same_row}, we get:
\begin{equation}
[\boldsymbol\Upsilon_{\T}]_{lj}
=
\begin{cases}
\displaystyle
2\widetilde{\sigma}_{l,\min}^2\,\eta_l,
& l=j,
\\[4mm]
\displaystyle
-
\frac{
\tilde{\sigma}_{l,\max}^2
\|\bar \H_{l,l}\|_2^2
\|\bar \H_{j,l}\|_2^2
}{
\sigma_{v_l}^4
\operatorname{tr}(\boldsymbol\Omega_l)^2
},
& l\neq j.
\end{cases}
\label{eq:Upsilon condensed entries}
\end{equation}
We can now invoke the monotone NEP framework in~\cite{scutari2012monotone} to establish convergence of the proposed semantic game.
\begin{theorem}
Consider the game
$\mathcal{G}_{\mathrm{sem}}$
and its equivalent VI reformulation
$\mathbf{VI}(\mathcal{\bar Q},\T)$.
If 
$\boldsymbol{\Upsilon}_{\mathbf{T}}$
in~\eqref{eq:Upsilon condensed entries}
is a $P$-matrix\footnote{
A matrix $\A\in\mathbb{R}^{n\times n}$ is a $P$-matrix if all its principal minors are positive.
},
then $\T$ is a uniformly $P$-function, the game admits a unique pure-strategy Nash equilibrium, and every sequence
$\{(\boldsymbol{\varphi}^{(t)}_l)_{l=1}^L\}_{t=1}^{T}$
generated by Algorithm~\ref{algo: potential game for SCE}
converges to the unique equilibrium
$\boldsymbol{\varphi}^{\star}$,
for any feasible updating scheme.
\label{thm:async_BR_convergence}
\end{theorem}
\begin{proof}
The result follows from~\cite[Th.~2]{scutari2012monotone}. In particular, the matrix
$\boldsymbol{\Upsilon}_{\mathbf T}$
provides a worst-case characterization of the diagonal dominance properties of the pseudo-gradient Jacobian. If
$\boldsymbol{\Upsilon}_{\mathbf T}$
is a $P$-matrix, then $\T$ is a uniformly $P$-function, which guarantees uniqueness of the Nash equilibrium and convergence of the distributed best-response dynamics.
\end{proof}

The condensed matrix in~\eqref{eq:Upsilon condensed entries} provides a tractable surrogate to study the monotonicity properties of $\T$. In particular, sufficient conditions ensuring that $\boldsymbol\Upsilon_T$ is a $P$-matrix directly yield convergence guarantees for the proposed ISWF dynamics. From (\ref{eq:Upsilon_Gamma_same_row}), a sufficient condition is the existence of
$\boldsymbol\tau=(\tau_1,\ldots,\tau_L)>\mathbf 0$
such that
\begin{equation}
\frac{1}{\tau_l}
\sum_{j\neq l}
\tau_j
\frac{
\beta_{lj}^{\max}
}{
\alpha_l^{\min}
}
<1,
\qquad
\forall\, l=1,\ldots,L,
\label{eq:weighted_dd}
\end{equation}
which corresponds to a weighted strict diagonal dominance condition on $\boldsymbol\Upsilon_T$. Under~\eqref{eq:weighted_dd}, Theorem~\ref{thm:async_BR_convergence} guarantees global convergence of ISWF to the unique Nash equilibrium of $\mathcal{G}_{\mathrm{sem}}$. Substituting~\eqref{eq:alpha_bound} and~\eqref{eq:beta_final_bound} into~\eqref{eq:weighted_dd}, and defining
\[
\kappa(\tilde{\boldsymbol\Sigma}_{p_l})
\triangleq
\frac{
\widetilde{\sigma}_{l,\max}
}{
\widetilde{\sigma}_{l,\min}
},
\]
we obtain the sufficient condition
\begin{equation}
\sum_{j\neq l}
\frac{\tau_j}{\tau_l}
\|\bar{\mathbf H}_{j,l}\|_2^2
<
\frac{
2\sigma_{v_l}^4
\operatorname{tr}(\boldsymbol\Omega_l)^2
}{
\kappa(\tilde{\boldsymbol\Sigma}_{p_l})^2
\|\bar{\mathbf H}_{l,l}\|_2^2
}
\,\eta_l,
\qquad
\forall\, l\in L.
\label{eq:weighted_dd_conditioning_final}
\end{equation}
\begin{remark}
Condition~\eqref{eq:weighted_dd_conditioning_final} explicitly links the convergence of the proposed ISWF dynamics to both the physical interference topology and the semantic alignment structure. In particular, strong interference and poorly conditioned semantic subspaces, i.e., large
$\kappa(\widetilde{\boldsymbol\Sigma}_{p_l})$,
shrink the convergence region and make the distributed dynamics harder to stabilize. This occurs when semantic alignment is concentrated over a few dominant latent modes, whereas well-conditioned semantic subspaces enlarge the set of operating regimes satisfying the diagonal dominance condition. Hence,~\eqref{eq:weighted_dd_conditioning_final} establishes an explicit connection between latent-space alignment geometry and the physical interference environment.
\end{remark}

\section{Numerical Results}\label{sec:Numerical_results}
This section provides numerical results
to assess the performance of the proposed cognitive games for latent space alignment, in a multi-user MIMO 
system, considering Rician flat-fading channels with rice factor $\!=\!1.5$, path loss exponent $\!=\!2.5$, and uniform linear array with $N_{T_l}\!=\!N_{R_l}\!=\!8$.
Semantic communication is then implemented by transmitting latent representations, and its effectiveness with respect to the operational requirements of the target application is evaluated on image classification and reconstruction. We consider, for the classification task, the CIFAR-10 dataset \!(
$32 \!\times\! 32$ RGB samples).
Among them, 42500 images were used for training and 10000 for testing, with classification across 10 labels.
Considered latent representations are produced by the backbone of pre-trained models chosen from the \href{https://pypi.org/project/timm/}{\textit{timm}} Python library.  
Precisely, we deploy three MIMO point-to-point links, corresponding to secondary users, operating at carrier frequency $f_c \!=\!3.5 \text{ GHz}$, where each intended transmitter–receiver separation is $20$m, and the distance between adjacent transmitters varies from $20$m to $1200$ m. 

\begin{figure}[t]
  \centering
\includegraphics[trim=0.30cm 0.41cm 0.15cm 0.25cm,clip,width=1.0\columnwidth]{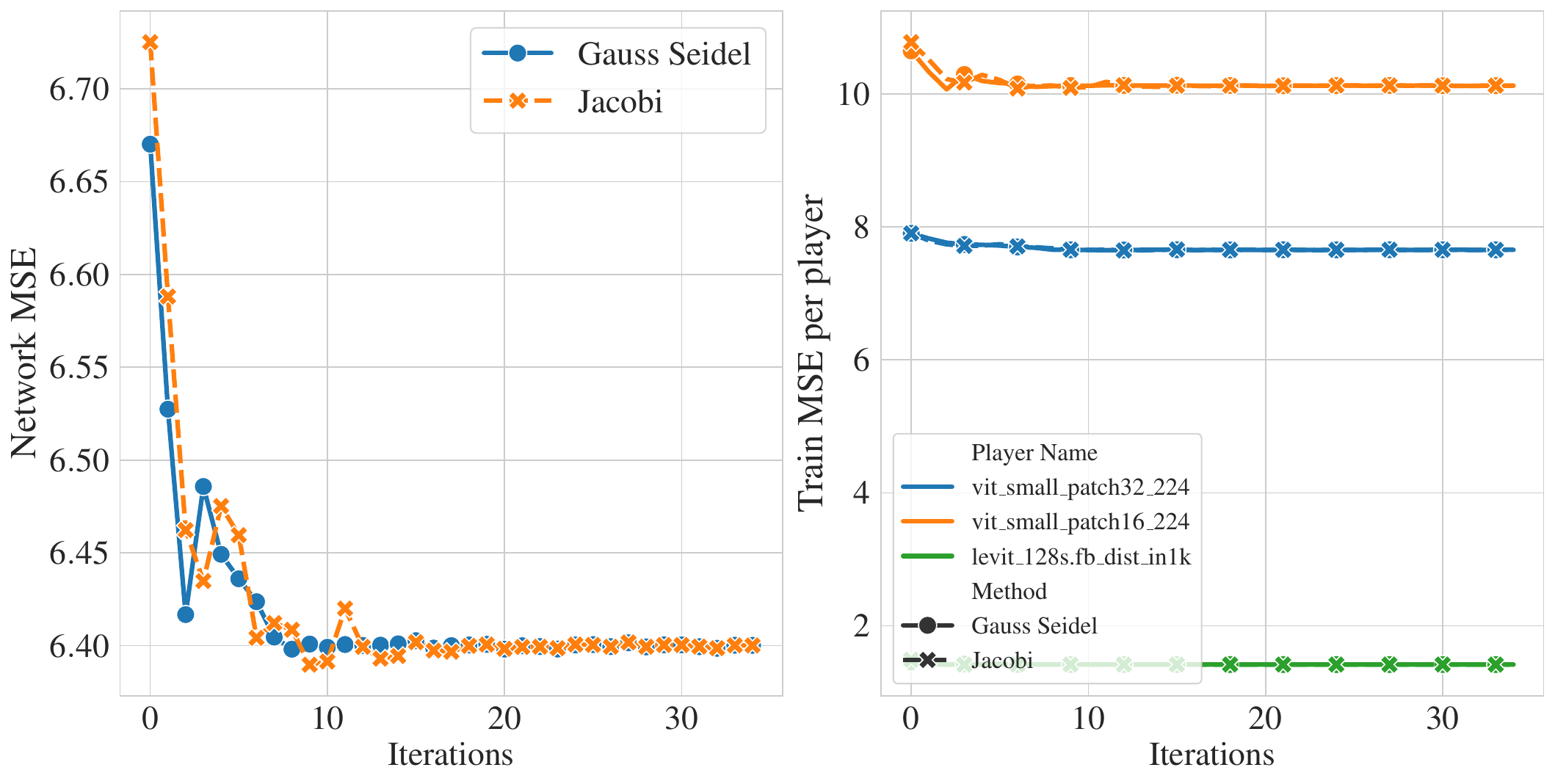}
  \caption{System BR-iterative behavior with $N_{T_l},N_{R_l}\!=\!8$, $K\!=\!1$, $\alpha^{SF}\!=\!6$. Network (left) and per player (right) $\text{MSE}_l(\boldsymbol{\Phi}_l)$ over game iterations:\textit{ vit\_small\_patch16} 
  is the tx deployed in the middle of the system, suffering the most MUI.}
  \label{fig::mse_game}
\end{figure}
Numerical simulations are evaluated by modeling the receiver
thermal noise with
$\sigma_{v_l}^2 \!=\! k_B T B$,
where $B$ denotes the noise-bandwidth, equal to $350$ MHz.
To modulate the strength of MUI term, let us define a \textit{MUI scaling factor} $\alpha^{SF} \!=\! \frac{d(T_i,T_j)}{d(T_i,R_i)}$  where the numerator represents the distance between the $i$-th transmitter and $j$-th interferer and the denominator the distance between the intended $i$-th transmitter-receiver. The scaling factor parametrizes the strength of the MUI term, i.e., the higher the scaling factor, the lower the  MUI term.
Unlicensed transmitters operate at average tx-power $\!P_{\max}\!=\!1\!$ \!W per channel use, with their internal logic given by the DNNs vit\_small\_patch32, vit\_small\_patch16, levit\_128s\.fb; receivers interpret latent codes using vit\_base\_patch16, vit\_tiny\_patch16, vit\_base\_patch32\_clip. All displayed results on task performance are averaged across multiple channel realizations,  
and the confidence bands reflect the variability among them.
For the sake of simplicity, we choose internal logic operating in transmission holding latent spaces of equal dimensions $d_l$, defining a common compression factor $\xi\!=\!
\!\frac{K}{d_l}$\! (cf. sec. \ref{sec: system_model}). In our experiments, we employed the recursive step-size update
$\gamma^{t+1}=\gamma^{t}(1-\epsilon\gamma^{t})$,
with 
$\epsilon \!=\! 10^{-3}$,
to improve the stability of the distributed dynamics. \\
\textbf{Convergence behavior.}
Figs.~\ref{fig::mse_game} and~\ref{fig::mui_game} illustrate the performance evolution of the proposed Iterative Semantic Water-Filling (ISWF) algorithm in terms of semantic alignment error, decoded MUI, and downstream task accuracy. Specifically, Fig.~\ref{fig::mse_game} reports the evolution of the MSE associated with the latent-space alignment objective, while Fig.~\ref{fig::mui_game} shows the corresponding decoded MUI power and task accuracy achieved during the distributed best-response dynamics. The generated sequence exhibits stable convergence across all users, progressively reducing both semantic distortion and multi-user interference. Interestingly, the communication link associated with the \texttt{vit\_small\_patch16} backbone (orange curve in Fig.~\ref{fig::mse_game}) experiences the largest alignment cost. This user is deployed in the central position of the network topology and is therefore subject to the strongest interference generated by the neighboring cross-links, effectively acting as the \emph{struggler} of the semantic game. Nevertheless, the proposed distributed strategy successfully adapts the semantic pre-equalizers to mitigate the interference generated toward this critical link, leading to a progressive reduction of the decoded MUI term for all users, as shown in Fig.~\ref{fig::mui_game}. At the same time, the downstream task accuracy progressively improves over the game iterations, confirming the effectiveness of the proposed latent-space alignment strategy under interference-limited conditions.
\begin{figure}[!t]
  \centering
\includegraphics[trim=0.30cm 0.41cm 0.05cm 0.25cm,clip,width=1.0\columnwidth]{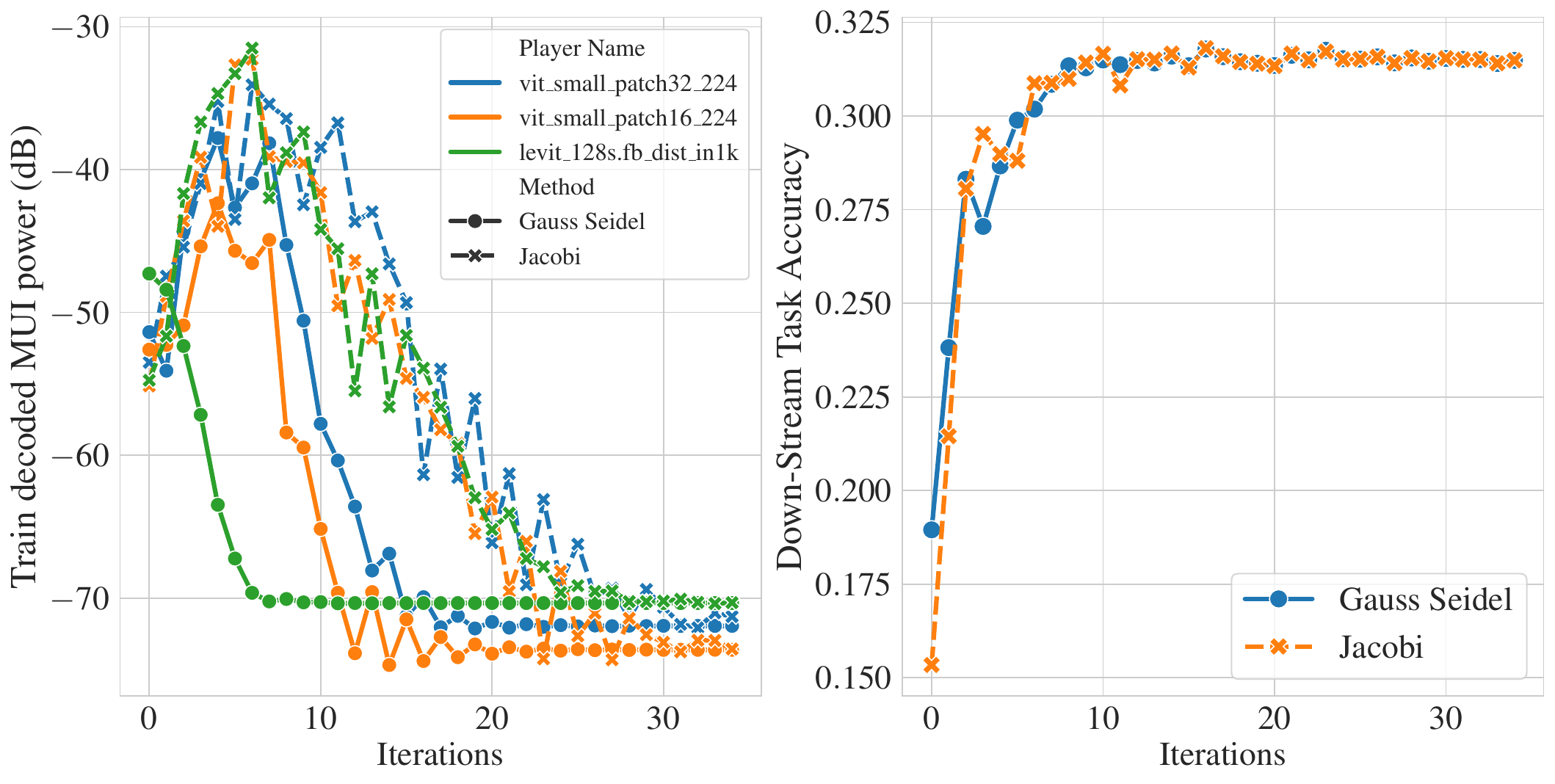}
  \caption{System BR-iterative behavior with $N_{T_l},N_{R_l}\!=\!8$, $K\!=\!1$, $\!\alpha^{SF}\!=\!6$. Decoded MUI power in dB \!(left)\! and task accuracy.}
  \label{fig::mui_game}
\end{figure}
\begin{figure*}[!t]
    \centering
    \includegraphics[
        width=\textwidth,
        trim={0.5cm 0.05cm 0.05cm 0.05cm},
        clip
    ]{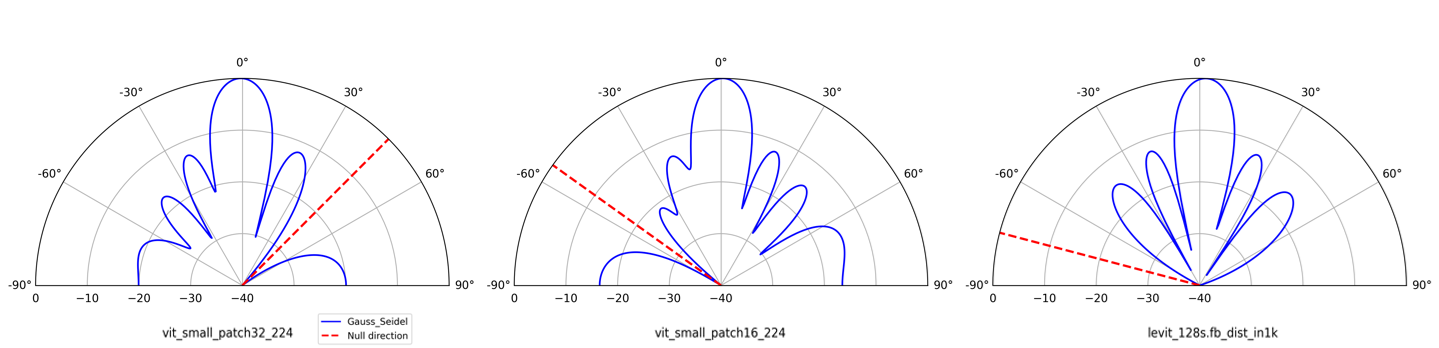}
    \caption{\!Transmitter radiation pattern of secondary user, under null interference and \!power constraint. Displayed allocated power for the sequential scheme with \!$K\!=\!1$, $\!N_{T_l}\!=\!8$ and $\!\alpha^{SF}\!=\!6$. The red dashed line identifies \!spatial direction toward primary user.}
    \label{fig:radiation_patterns}
\end{figure*}
These numerical evaluations highlight, for both update schemes, convergence
of the generated sequence toward the fixed point of the best-response dynamics defined in~\eqref{eq:BR dynamics}.
\begin{figure}[!tb]
    \centering
    \includegraphics[trim=0.30cm 0.32cm 0.15cm 0.25cm,clip,width=\columnwidth]{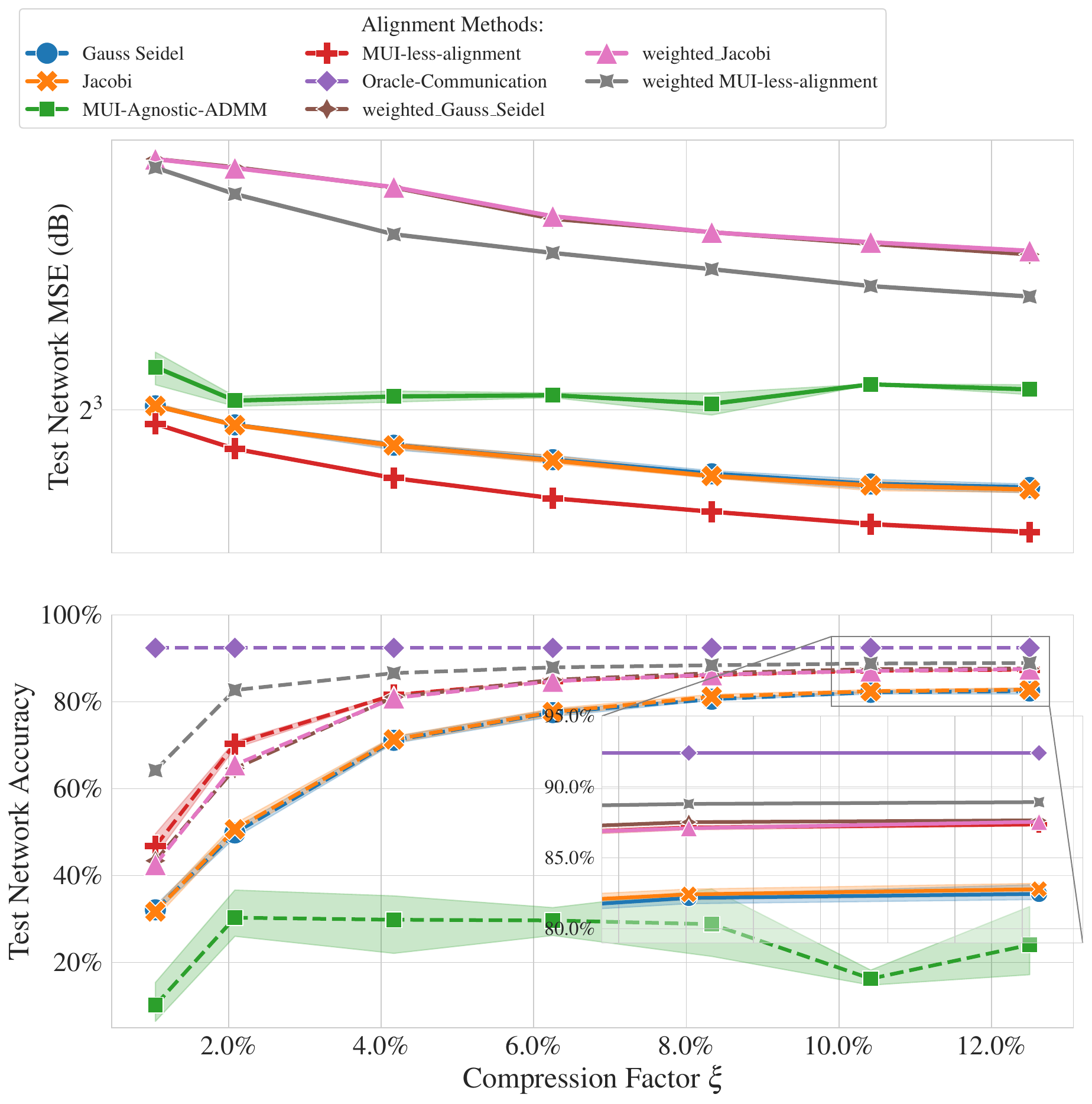}
    \caption{\small Task Accuracy-MSE vs. $\xi$, with $N_{T_l}\! =\! N_{R_l} \!=\! 8$ and $\alpha^{SF}\!=\!6$.}
\label{fig:compression_fact_classification}
\end{figure}

A physical interpretation of the learned semantic pre-equalizers is provided in Fig.~\ref{fig:radiation_patterns}, which depicts the transmitter radiation patterns obtained under null-interference and power constraints. The reformulation introduced in~\eqref{eq: equivalent null_constraint} effectively enables the identification of spatial null directions toward the protected primary user, highlighted by the red dashed line. The displayed patterns are generated under a sufficiently large Rice factor, such that the resulting channels are close to a geometric propagation regime. Interestingly, the proposed ISWF dynamics simultaneously learns to preserve the protected spatial directions while forming approximately broadside beam patterns toward the intended receivers. This behavior reflects the joint exploitation of semantic and physical channel subspaces in the distributed power allocation process.

\noindent \textbf{Task performance.}
To assess the proposed cognitive game-theoretic framework, we compare it against:
(i) \textit{MUI-less Alignment}, an ideal interference-free semantic alignment strategy obtained by neglecting the MUI contribution and setting $\R_n=\R_v$;
(ii) \textit{MUI-agnostic ADMM}, corresponding to the alternating optimization framework in~\cite{pandolfo2025latent}, where semantic MIMO transceivers are designed without cognitive interference management or primary-user protection constraints; and
(iii) \textit{Oracle-Communication}, representing the ideal upper bound corresponding to perfect semantic transmission without channel distortion. The considered schemes are evaluated on both classification and reconstruction tasks as functions of the semantic compression factor $\xi$ and the MUI scaling factor $\alpha^{SF}$, respectively in Figs.~\ref{fig:compression_fact_classification},~\ref{fig:scaling_fact_classification},~\ref{fig:compression_fact_reconstruction}, and~\ref{fig:reconstruction_example}. For the classification task, the weighting matrix $\boldsymbol\Omega_l$ in~\eqref{eq:nonconvexERM} is designed using entropy-based sample weights:
\begin{equation}
\omega_{i,l}
\triangleq
\exp\!\bigl(
\zeta\,\mathcal E(\y_{i,l})
\bigr),
\qquad
i\in\mathcal T_{r_l},
\quad
\forall l\in L,
\label{eq:entropy weights}
\end{equation}
where $\mathcal E(\y_{i,l})\in[0,1]$ denotes the normalized entropy of the downstream classifier and $\zeta$ controls the weighting strength. The rationale is to give more importance to samples on which classification confidence is lower. In the figures, the corresponding weighted schemes are labeled as \textit{Weighted}, whereas the unweighted ISWF implementations are reported as \textit{Jacobi} and \textit{Gauss--Seidel}. For the reconstruction task, we simply set $\boldsymbol\Omega_l=\I$. which corresponds to the unweighted case.

Fig.~\ref{fig:compression_fact_classification} reports latent-space MSE and classification accuracy versus the compression factor $\xi$. As $\xi$ decreases, stronger compression progressively degrades the transmitted latent information. Nevertheless, the proposed framework preserves robust task performance over a broad range of compression levels, significantly outperforming MUI-agnostic aligners and approaching the MUI-less benchmark as the number of channel uses increases. Interestingly, entropy-based weighting slightly increases latent-space MSE while improving classification accuracy, by emphasizing semantically difficult samples during alignment. This behavior is further confirmed in Fig.~\ref{fig:scaling_fact_classification}, which reports classification accuracy as a function of the MUI scaling factor $\alpha^{SF}$ under different antenna configurations. As interfering links get closer (smaller $\alpha^{SF}$), the proposed semantic game maintains reliable task execution, whereas interference-unaware schemes experience severe degradation due to the inability to suppress the MUI contribution.

Finally, Fig.~\ref{fig:compression_fact_reconstruction} reports MNIST reconstruction performance in terms of MSE and per-pixel reconstruction error, while Fig.~\ref{fig:reconstruction_example} shows representative reconstruction examples. Even in this simple setting, the proposed semantic game substantially improves reconstruction quality compared to interference-agnostic approaches, despite simultaneously enforcing spatial protection constraints toward the primary users.

\begin{figure}[!tb]
    \centering
    \includegraphics[width=\columnwidth]{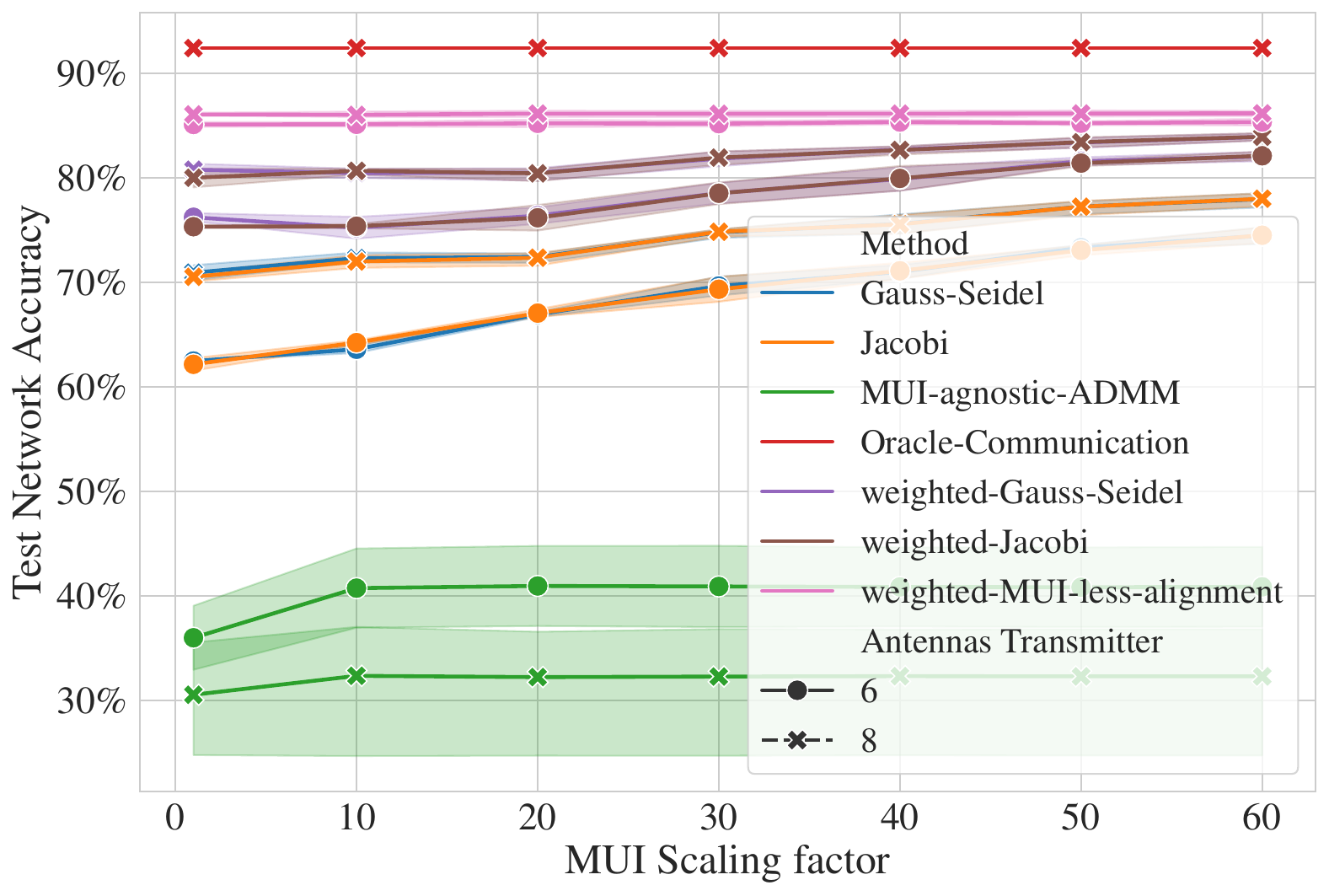}
    \caption{\small Network classification Accuracy vs. $\alpha^{SF}$ (MUI scaling factor), with $N_{T_l} = N_{R_l} \!=\! 6,8$ and wireless channel usage $K\!=\!4$.}
\label{fig:scaling_fact_classification}
\end{figure}

\begin{figure}[!t]
    \centering    \includegraphics[width=0.5\textwidth]{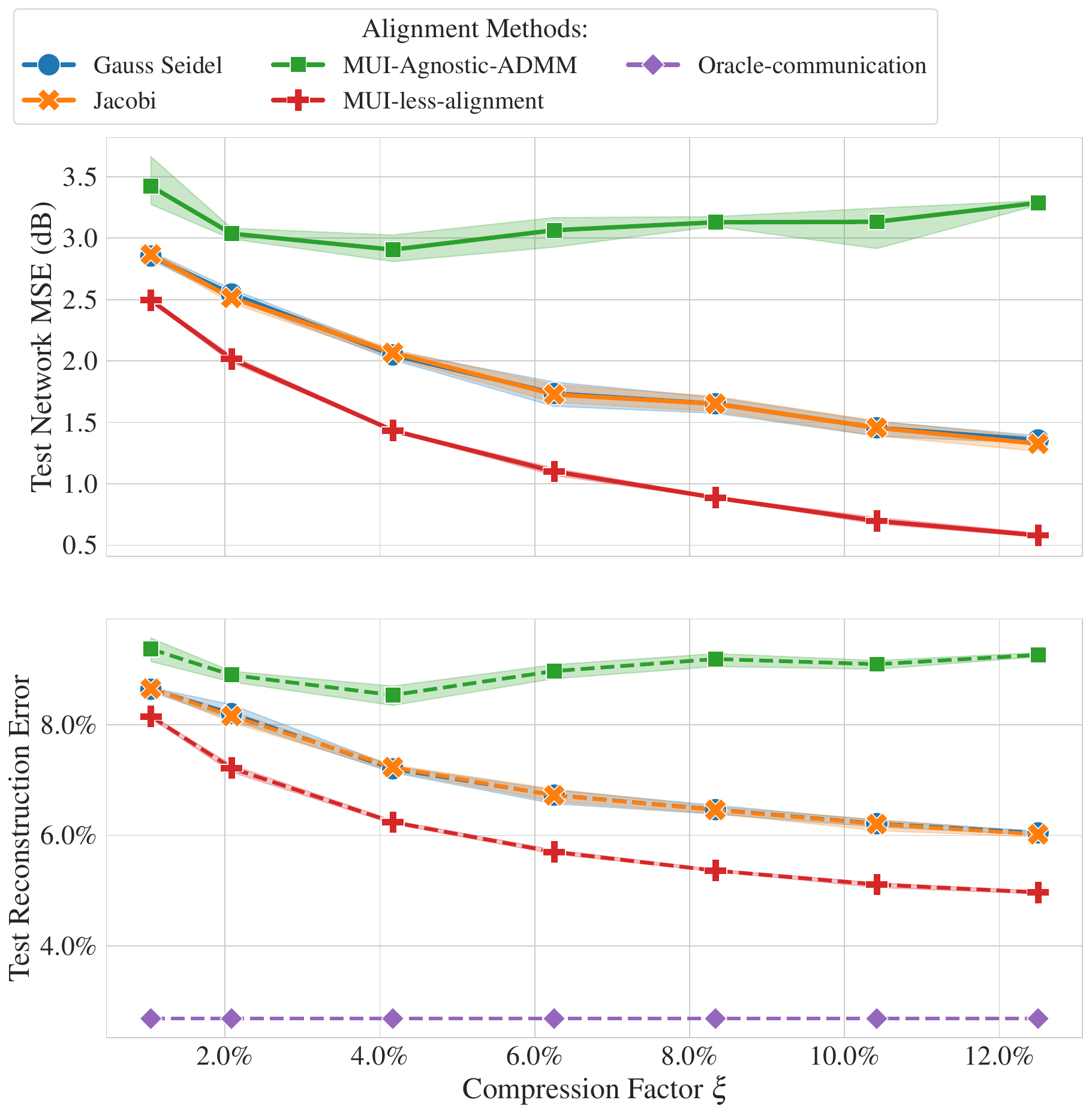}\caption{Per-pixel reconstruction error and MSE vs. $\xi$, considering MNIST dataset with $N_{T_l}\! =\! N_{R_l} \!=\! 8$ and $\alpha^{SF}\!=\!6$.}
\label{fig:compression_fact_reconstruction}
\end{figure}
\noindent \textbf{Resilience to Imperfect CSI.}
Fig.\ref{fig:CSI_estimation resilience} provides 
classification performance of Algorithm \eqref{algo: potential game for SCE} with imperfect CSI, relaxing Assumption A2. The estimated channel is modeled as in \cite{tejerina2025sum}, such that
$\hat{\H}\!=\!\sqrt{1\!-\!\rho^2}\mathbf{\H}^{(0)}\!+\! \rho\tilde{\H}$,
where ${\H}^{(0)}\!$ \!is the initial known channel, $\tilde\H$ denotes an independent Gaussian channel estimation error and \!$\rho\!\in\![0,1]\!$ quantifies the estimation accuracy. From Fig. Fig.\ref{fig:CSI_estimation resilience}, we can notice that, also with high estimation error ($\rho\!\approx\!0.8$), the performance degradation is tolerable. This further highlights the effectiveness of the proposed design for semantic transceivers, which resides in the joint exploitation of physical and semantic channel \eqref{eq:F_definition}.
\begin{figure}[!t]
    \centering
\includegraphics[width=0.5\textwidth]{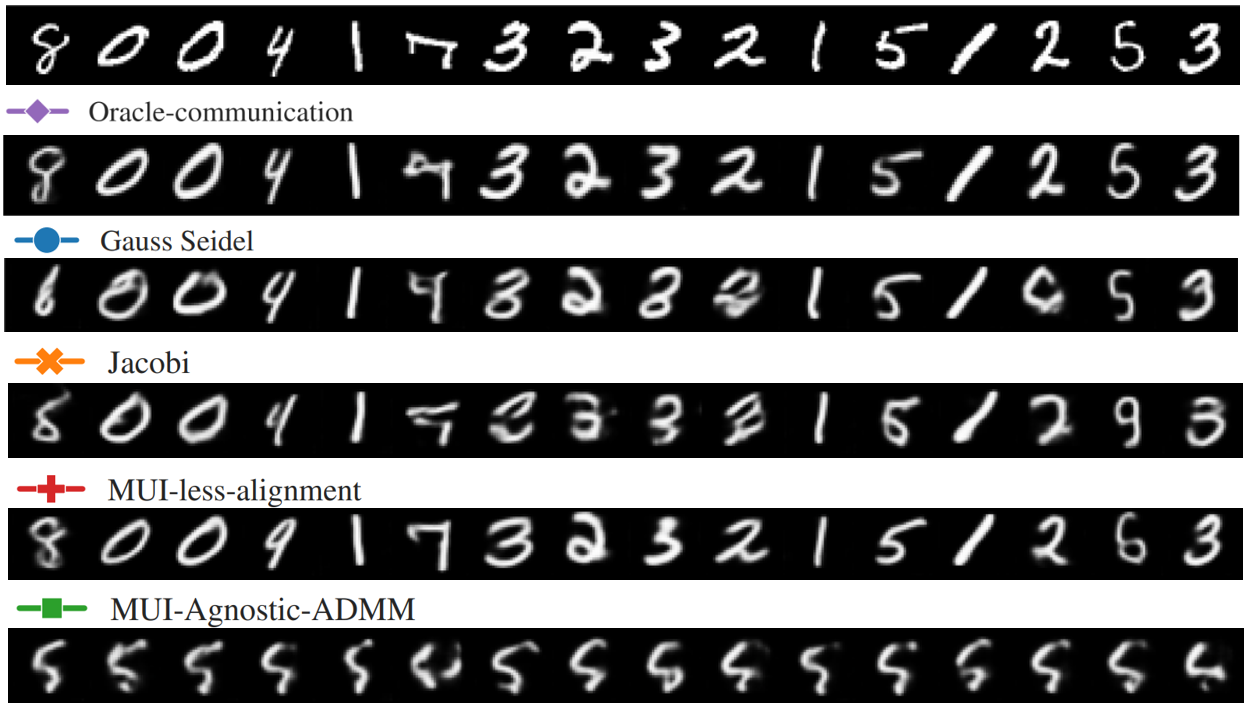}\caption{Reconstruction task example on MNIST, considering $K\!=\!12$ and $\!\alpha^{SF}\!=\!6$ at the struggler player (vit\_tiny\_patch\_16 receiver). The first row represents the target test set.
    }
\label{fig:reconstruction_example}
\end{figure}

\section{Conclusions}
\label{sec:conclusions}
In this paper, we proposed a novel framework for cognitive semantic communications that jointly addresses latent-space misalignment and semantic coexistence under interference and power constraints, through end-to-end optimized MIMO semantic transceivers. Restricting attention to linear transformations, we derived the optimal semantic equalizers up to a best rank approximation over the semantic alignment subspaces.
A key outcome of this work is the characterization of the semantic channel equalization problem over interference channels through a semantic water-filling structure, which enables the original matrix-valued latent-space alignment problem to be recast as a lower-dimensional vector power-control game. Building on this reformulation, we provided a complete game-theoretic characterization of the proposed semantic interaction, deriving sufficient conditions for existence, uniqueness, and convergence of the proposed Iterative Semantic Water-Filling algorithm. Extensive numerical results validate the effectiveness of the proposed framework, highlighting the interplay among semantic compression, interference mitigation, and latent-space alignment in AI-native wireless systems.

\begin{figure}[!t]
    \centering
\includegraphics[trim=0.30cm 0.38cm 0.25cm 0.25cm,clip,width=\columnwidth]{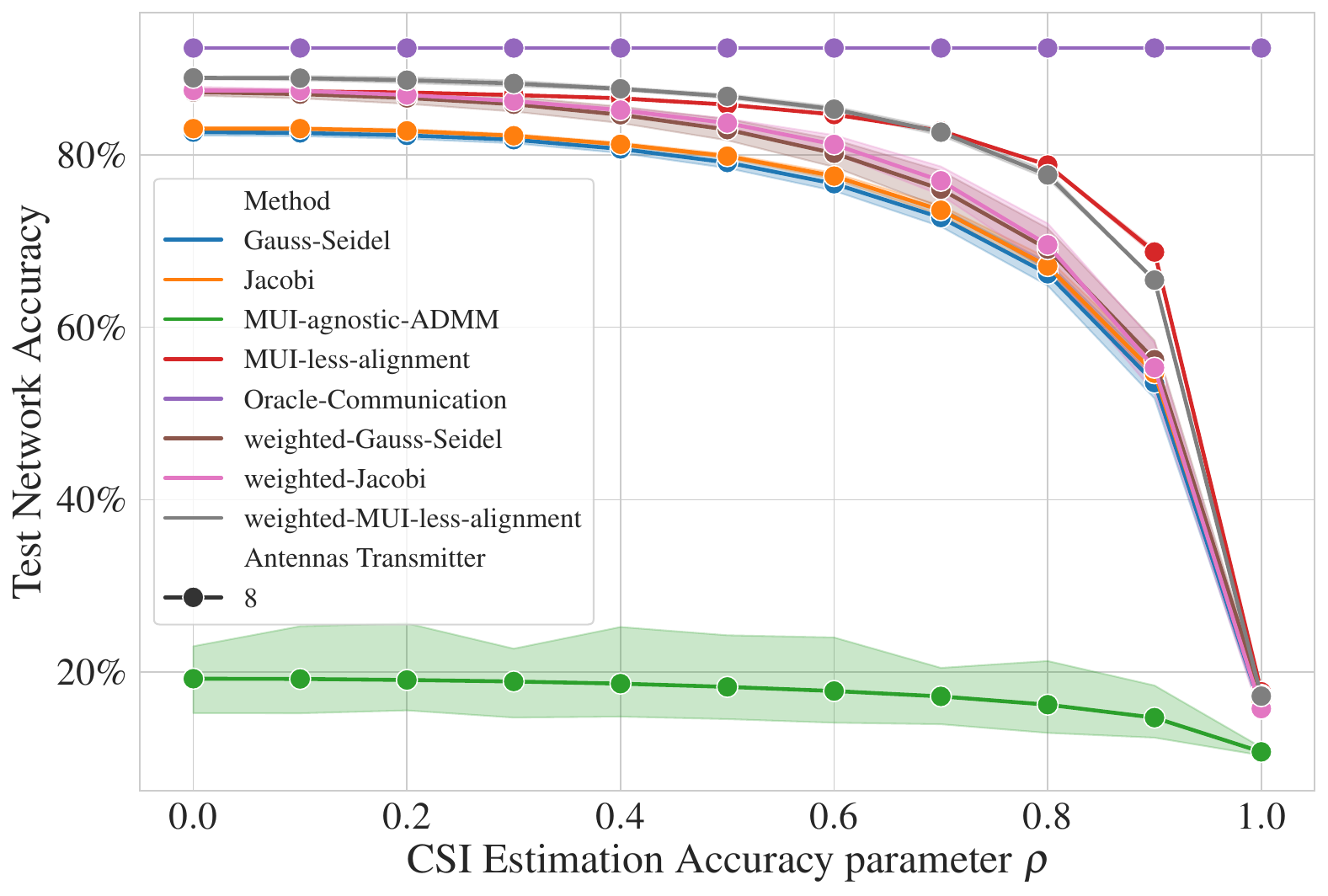}\caption{Resilience of Algorithm \ref{algo: potential game for SCE} to imperfect CSI estimation, with $K\!=\!12, \alpha^{SF}\!=\!6, N_{T_l}\!=\!8$.}
\label{fig:CSI_estimation resilience}
\end{figure}

\appendices
\section{}
 \label{appA} 
 By simply plugging 
 in the closed-form solution of $\G_l^{opt}$ by 
$(\ref{eq:Wiener_filter})$ in the objective function of problem $(\ref{eq:nonconvexERM})$, the mean square error with respect to $\F_l$ can be expressed as:
\vspace{-0.15 cm}
 \begin{equation}
    \text{MSE}_l(\F_l) \!=\! 
    \text{tr} \{ \R_{y_l}^{(\boldsymbol\Omega)} \!-\!  \P_l \left[ \S_l^H (\R_{n_l} \!+\! \S_l \S_l^H)^{-1} \S_l \right] \P_l^H \}
    \label{eq:mse_F_step1}
 \end{equation}
 with $\P_l\!=\!\Y_l \boldsymbol{\Omega}_l \X_l^H$ and $\S_l \!=\! \H_{l, 
 l}\F_l$.
 Expanding by Woodbury identity, the inversion term reads as:
 \begin{equation}
     (\R_{n_l} + \S_l \S_l^H)^{-1} \!=\! \R_{n_l}^{-1} \!-\! \R_{n_l}^{-1} \S_l (\I \!+\! \S_l^H \R_{n_l}^{\!-\!1} \S_l)^{\!-\!1} \S_l^H\R_{n_l}^{\!-\!1}
 \end{equation}
 and by Searle Set of identities \cite{petersen2008matrix}, 
 equation (\ref{eq:mse_F_step1}) becomes:
 \begin{equation}
     \text{MSE}_l(\F_l)\! =\! 
     \text{tr} \{ \R_{y_l}^{(\boldsymbol\Omega)} \!-\!
      \P_l \P_l^H + \P_l(\F_l^H \R_{H_l} \F_l + \I)^{\!-\!1} \P_l^H\}. \nonumber
\label{eq:objective_reformulated}
 \end{equation}

 \section{}
\label{app: Lemma Spectral Bound}

From~\eqref{eq:F_definition}, we have
$\bar\F_j\bar\F_j^H
=
\sum_{m=1}^{\overline{KN_{T_j}}}
\varphi_{j,m}\,
\boldsymbol\nu_{h_{j,m}}
\boldsymbol\nu_{h_{j,m}}^H,$
where $\boldsymbol\nu_{h_{j,m}}$ denotes the $m$-th column of $\V_{h_j}$. Since
$\{\boldsymbol\nu_{h_{j,m}}\}_m$
is an orthonormal set,
$\boldsymbol\nu_{h_{j,m}}\boldsymbol\nu_{h_{j,m}}^H\preceq\I$
for all $m$, and therefore
$\bar\F_j\bar\F_j^H
\preceq
\left(
\sum_m \varphi_{j,m}
\right)\I.$ Using the power constraint in~\eqref{eq:power-opt},
$\sum_m \varphi_{j,m}\le KP_{\max}$,
we obtain
\begin{equation}\label{eq:bound1}
\bar\F_j\bar\F_j^H \preceq KP_{\max}\I.    
\end{equation}
Substituting (\ref{eq:bound1}) into the reduced MUIN covariance matrix in~\eqref{eq:MUI_covariance_reduced} yields
\[
\sigma_{v_l}^2
\operatorname{tr}(\boldsymbol\Omega_l)\I
\preceq
\bar\R_{n_l}
\preceq
KP_{\max}
\sum_{j\neq l}
\bar\H_{j,l}\bar\H_{j,l}^H
+
\sigma_{v_l}^2
\operatorname{tr}(\boldsymbol\Omega_l)\I.
\]
Defining $\hat\R_{n_l}$ in \eqref{eq:Rnhat},
we have
$
\sigma_{v_l}^2
\operatorname{tr}(\boldsymbol\Omega_l)\I
\preceq
\bar\R_{n_l}
\preceq
\hat\R_{n_l}.
$
Since matrix inversion reverses the Loewner order over positive definite matrices, i.e.,
$\A\preceq\B \Rightarrow \B^{-1}\preceq\A^{-1}$,
we get
\begin{equation}\label{eq:bound2}
\hat\R_{n_l}^{-1}
\preceq
\bar\R_{n_l}^{-1}
\preceq
\big(
\sigma_{v_l}^2
\operatorname{tr}(\boldsymbol\Omega_l)
\big)^{-1}\I.    
\end{equation}
Now, recalling the effective channel covariance matrix
\begin{equation}\label{eq:effective_channel_covariance}
\bar\R_{H_l}
=
\bar\H_{l,l}^H
\bar\R_{n_l}^{-1}
\bar\H_{l,l},
\end{equation}
premultiplying and postmultiplying~\eqref{eq:bound2} by
$\bar\H_{l,l}^H$ and $\bar\H_{l,l}$, respectively, gives
\[
\bar\H_{l,l}^H
\hat\R_{n_l}^{-1}
\bar\H_{l,l}
\preceq
\bar\R_{H_l}
\preceq
\big(
\sigma_{v_l}^2
\operatorname{tr}(\boldsymbol\Omega_l)
\big)^{-1}
\bar\H_{l,l}^H\bar\H_{l,l}.
\]
The eigenvalue bounds in~\eqref{eq:R_H spectral bound} then follow by monotonicity of the eigenvalues under Loewner ordering~\cite[Cor.~9.9]{tropp2022acm}.

Next, we derive the eigenvalue sensitivity bound in ~\eqref{eq: sensitivity bound Lemma}. Differentiating~\eqref{eq:effective_channel_covariance} with respect to $\varphi_{j,n}$, and using
$\partial \bar\R_{n_l}^{-1}
=
-\bar\R_{n_l}^{-1}
(\partial \bar\R_{n_l})
\bar\R_{n_l}^{-1}$,
yields
\begin{equation}
\frac{\partial \bar\R_{H_l}}
{\partial \varphi_{j,n}}
=
-
\bar\H_{l,l}^H
\bar\R_{n_l}^{-1}
\left(
\frac{\partial \bar\R_{n_l}}
{\partial \varphi_{j,n}}
\right)
\bar\R_{n_l}^{-1}
\bar\H_{l,l}.
\label{eq:dAl_appendix}
\end{equation}
By eigenvalue perturbation theory~\cite[Eq.~67]{petersen2008matrix}, the derivative of the $m$-th eigenvalue of $\bar\R_{H_l}$ with respect to $\varphi_{j,n}$ is
\begin{equation}\label{eq:delambda}
\frac{\partial \lambda_{l,m}}
{\partial \varphi_{j,n}}
=
\boldsymbol\nu_{l,m}^H
\frac{\partial \bar\R_{H_l}}
{\partial \varphi_{j,n}}
\boldsymbol\nu_{l,m},
\end{equation}
where $\boldsymbol\nu_{l,m}$ is the eigenvector associated with $\lambda_{l,m}$. Substituting~\eqref{eq:dAl_appendix} into~\eqref{eq:delambda} gives
\[
\frac{\partial \lambda_{l,m}}
{\partial \varphi_{j,n}}
=
-
\boldsymbol\nu_{l,m}^H
\bar\H_{l,l}^H
\bar\R_{n_l}^{-1}
\left(
\frac{\partial \bar\R_{n_l}}
{\partial \varphi_{j,n}}
\right)
\bar\R_{n_l}^{-1}
\bar\H_{l,l}
\boldsymbol\nu_{l,m}.
\]
Combining~\eqref{eq:delambda} and~\eqref{eq:dAl_appendix} with
$|\a^H\B\a|\le\|\B\|_2\|\a\|_2^2$
yields
\begin{equation}
\left|
\frac{\partial \lambda_{l,m}}
{\partial \varphi_{j,n}}
\right|
\le
\|\bar\H_{l,l}\|_2^2
\|\bar\R_{n_l}^{-1}\|_2^2
\left\|
\frac{\partial \bar\R_{n_l}}
{\partial \varphi_{j,n}}
\right\|_2.
\label{eq:normbound_appendix}
\end{equation}
Moreover, differentiating~\eqref{eq:MUI_covariance_reduced} gives
\begin{equation}
\frac{\partial \bar\R_{n_l}}
{\partial \varphi_{j,n}}
=
\bar\H_{j,l}
\boldsymbol\nu_{h_j,n}
\boldsymbol\nu_{h_j,n}^H
\bar\H_{j,l}^H,
\label{eq:derivativeRn}
\end{equation}
and therefore
\begin{equation}
\left\|
\frac{\partial \bar\R_{n_l}}
{\partial \varphi_{j,n}}
\right\|_2
=
\|
\bar\H_{j,l}\boldsymbol\nu_{h_j,n}
\|_2^2
\le
\|\bar\H_{j,l}\|_2^2.
\label{eq:boundRn}
\end{equation}
Furthermore, from~\eqref{eq:bound2},
\begin{equation}
\|\bar\R_{n_l}^{-1}\|_2
\le
\frac{1}{
\sigma_{v_l}^2
\operatorname{tr}(\boldsymbol\Omega_l)
}.
\label{eq:boundRninv}
\end{equation}
Finally, substituting~\eqref{eq:boundRn} and~\eqref{eq:boundRninv} into~\eqref{eq:normbound_appendix} yields
\[
\left|
\frac{\partial \lambda_{l,m}}
{\partial \varphi_{j,n}}
\right|
\le
\frac{
\|\bar\H_{l,l}\|_2^2
\|\bar\H_{j,l}\|_2^2
}{
\sigma_{v_l}^4
\operatorname{tr}(\boldsymbol\Omega_l)^2
},
\]
which proves~\eqref{eq: sensitivity bound Lemma}.

\balance
\bibliographystyle{IEEEtran}
\bibliography{ref}
\end{document}